\documentclass[12pt, draftclsnofoot, onecolumn]{IEEEtran}
\usepackage{amsthm}
\DeclareMathAlphabet{\mathpzc}{OT1}{pzc}{m}{it}
\usepackage{mathrsfs}
\usepackage{amssymb}
\usepackage{color}
\usepackage{amsmath}
\usepackage{amsfonts}
\usepackage{graphicx}
\usepackage{setspace}
\usepackage{algorithm}
\usepackage{algorithmic}
\usepackage{epstopdf}
\usepackage{cite}
\usepackage{hyperref}
\usepackage{subfigure}
\usepackage{bm}
\usepackage{fancyhdr}
\usepackage{etoolbox}
\usepackage{multirow}
\usepackage{booktabs}
\makeatletter
\patchcmd{\@makecaption}
{\scshape}
{}
{}
{}
\makeatletter
\patchcmd{\@makecaption} {\\}
{.\ }
{}
{}
\makeatother

\newtheorem{theorem}{\textbf{Theorem}}

\newtheorem{corollary}{\textbf{Corollary}}

\newtheorem{definition}{\textbf{Definition}}

\newtheorem{lemma}{\textbf{Lemma}}

\newtheorem{remark}{\textbf{Remark}}

\begin{document}
\abovedisplayskip=2.2pt
\belowdisplayskip=2.2pt
\begin{spacing}{1.0}
\title{Coded Computing and Cooperative Transmission for Wireless Distributed Matrix Multiplication} \vspace{-5mm}
\author{\IEEEauthorblockN{Kuikui Li, Meixia Tao, Jingjing Zhang, and Osvaldo Simeone }
\thanks{This work was presented in part at the IEEE ISIT 2020 \cite{kklisit}. Kuikui Li and Meixia Tao are with the Department of Electronic Engineering, Shanghai Jiao Tong University, Shanghai, P. R. China (Email: \{kuikuili, mxtao\}@sjtu.edu.cn). Jingjing Zhang and Osvaldo Simeone are with the KCLIP Lab, Department of Engineering, King{'}s College London, London, UK (Email: \{jingjing.1.zhang, osvaldo.simeone\}@kcl.ac.uk).

The work by K. Li and M. Tao is supported by the National Key R$\&$D Project of China under grant 2020YFB1406802 and the National Natural Science Foundation of China under grant 61941106. The work by J. Zhang and O. Simeone is supported by the European Research Council under the European Union's Horizon 2020 Research and Innovation Programme (Grant Agreement No. 725731).}}
\maketitle
\vspace{-4mm}
\begin{abstract}
\vspace{-2mm}
Consider a multi-cell mobile edge computing network, in which each user wishes to compute the product of a user-generated data matrix with a network-stored matrix. This is done through task offloading by means of input uploading, distributed computing at edge nodes (ENs), and output downloading. Task offloading may suffer long delay since servers at some ENs may be straggling due to random computation time, and wireless channels may experience severe fading and interference. This paper aims to investigate the interplay among upload, computation, and download latencies during the offloading process in the high signal-to-noise ratio regime from an information-theoretic perspective. A policy based on cascaded coded computing and on coordinated and cooperative interference management in uplink and downlink is proposed and proved to be approximately optimal for a sufficiently large upload time. By investing more time in uplink transmission, the policy creates data redundancy at the ENs, which can reduce the computation time, by enabling the use of coded computing, as well as the download time via transmitter cooperation. Moreover, the policy allows computation time to be traded for download time. Numerical examples demonstrate that the proposed policy can improve over existing schemes by significantly reducing the end-to-end execution time.
\end{abstract}
\begin{IEEEkeywords}
\vspace{-2mm}
Matrix Multiplication, Straggler, Edge Computing, Transmission Cooperation, Coded Computing \end{IEEEkeywords}
\section{Introduction}
\textbf{Motivation and scope:} Mobile edge computing (MEC) is an emerging network architecture that enables cloud-computing capabilities at the edge nodes (ENs) of mobile networks\cite{6923537,MEC,7931566}. Through task offloading, MEC makes it possible to offer mobile users intelligent applications, such as recommendation systems or gaming services, that would otherwise require excessive on-device storage and computing resources. Deploying task offloading, however, poses non-trivial design problems. On one hand, task offloading may require a large amount of data to be transferred between users and ENs over uplink or downlink channels, which may suffer severe channel fading and interference conditions, resulting in large communication latencies.  On the other hand, edge servers are likely to suffer from the straggling effect, yielding unpredictable computation delays\cite{dean2013tail}. A key problem in MEC networks, which is the subject of this paper, is to understand the interplay and performance trade-offs between two-way communication (in both uplink and downlink) and computation during the offloading process.

To this end, this study focuses on the baseline problem of computing the product between user-generated data vectors $\{\mathbf{u}\}$ and a network-stored matrix $\mathbf{A}$.  Matrix multiplication is a representative computation task that underlies many machine learning and data analytic problems. Examples of applications include recommendation systems based on collaborative filtering\cite{5197422}, in which the user-generated data $\{\mathbf{u}\}$ corresponds to user profile vectors, while the network-side matrix $\mathbf{A}$ collects the profile vectors of a certain class of items, e.g., movies. Matrix $\mathbf{A}$ is generally very large in practice, preventing a simple solution whereby users download and store the matrix for local computation.

Matrix multiplication, as many other more complex computations\cite{bekkerman2011scaling}, can be decomposed into subtasks and distributedly computed across multiple servers.
In MEC networks, the servers are embedded in distinct ENs, and hence distributed computing at the edge requires input data uploading via the uplink, computation at the ENs, and output data downloading via the downlink. A fundamental question that this work tackles is: \emph{What is the minimum achievable upload-compute-download latency triplet for completing matrix multiplication in the presence of straggling servers and multi-cell interference?}

In the task offloading process discussed above, the overall latency is the sum of three components, namely the time needed for input uploading, server computing, and output downloading. This paper is devoted to studying the interplay and trade-offs among these three components from an information-theoretic standpoint. A key result that will be illustrated by our results is that investing more time in any one of the three steps may be instrumental in reducing the time needed for subsequent steps thanks to coded computing \cite{fundamental_tradeoff,Speeding_Up,songzestraggler,improvedJingjing,8502151,jingjing,8006960,stagglerexploit,alternativetrade-off,RandomConnectivity,YanTradeoff,YanStraggling} and cooperative transmission \cite{7857805,sengupta2017fog,FanXu,8624603,TaoCache}. As explained next, both coded computing and cooperative transmissions leverage forms of computation redundancy.
\begin{figure}[t]
\centering
\includegraphics[width=4.6in, height=2.5in]{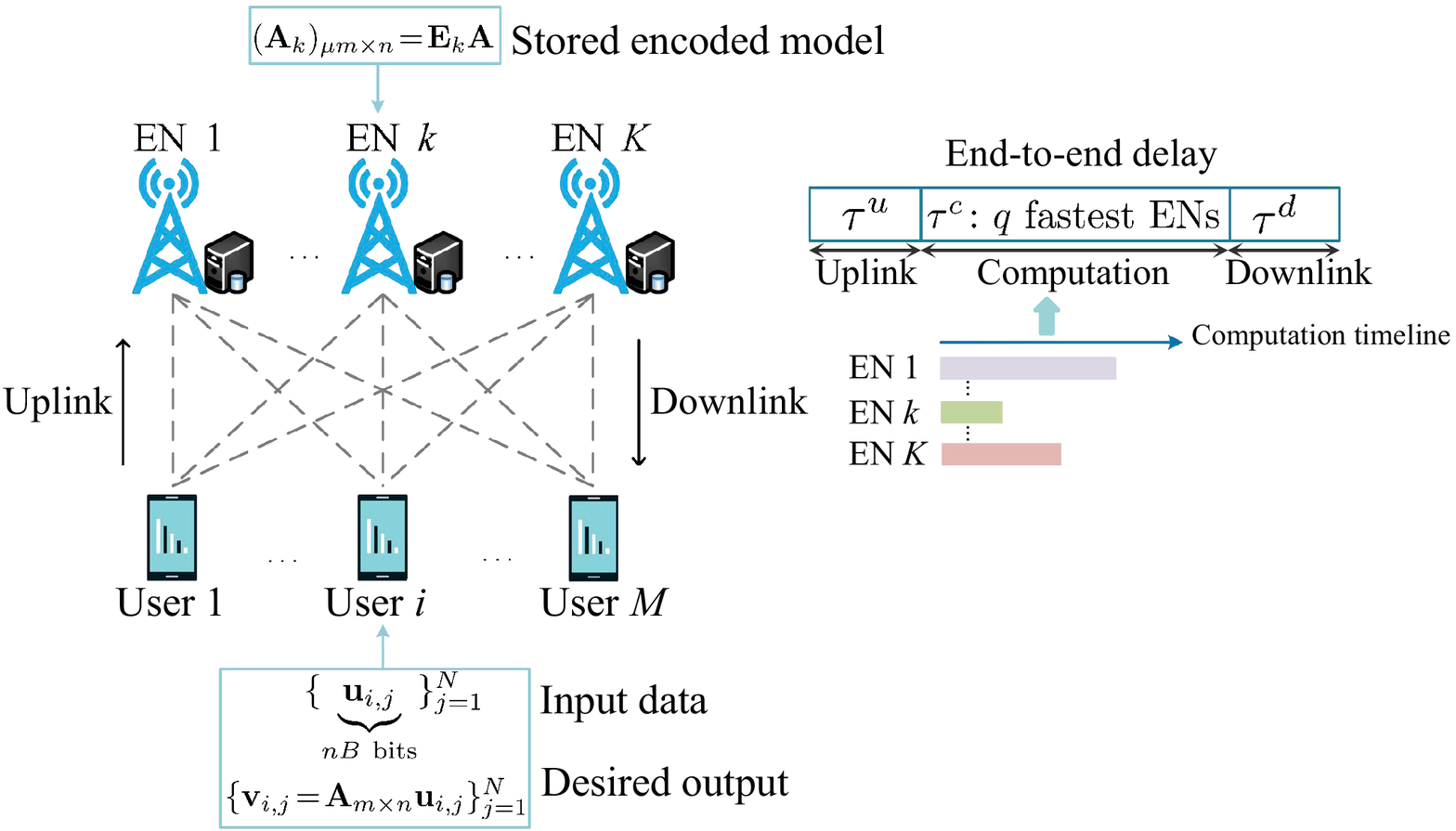}
\vspace{-2mm}
\caption{A multi-cell MEC network carrying out distributed matrix multiplication via uplink communication, edge computing, and downlink communications.}
\vspace{-8mm}
\label{sysModel}
\end{figure}

\textbf{Background and related works:} Coded computing was introduced in \cite{Speeding_Up} for a master-slave system with ideal communication links and linear computations. The approach aims at reducing the average latency caused by distributed servers with random computation time, hence mitigating the problem of \emph{straggling servers}\cite{dean2013tail}, through linear coding of the rows of matrix $\mathbf{A}$. Linear coding assigns each server a flexible number of encoded rows of matrix $\mathbf{A}$. Thanks to maximum distance separable (MDS) coding, assigning more coded rows at the servers reduces the number of servers that need to complete their computations in order to recover the desired outputs\cite{songzestraggler,improvedJingjing,8502151}. Coded computing was introduced in \cite{jingjing} as a means to speed up the computation of distributed matrix multiplication in a MEC system, providing a starting point for this work.

A simple way to ensure computation redundancy is to assign repeatedly the same rows of matrix $\mathbf{A}$ across multiple ENs. While this does not provide the same robustness against stragglers as MDS coding, it allows ENs to compute common outputs, i.e., computation replication, as proposed in \cite{KKLMEC}. This in turn makes it possible for the ENs to cooperate for transmission to the users in the downlink, which can reduce the download latency in an interference-limited system such as multi-user multi-server MEC systems shown in Fig. \ref{sysModel}. This form of cooperative transmission enabled by computation redundancy has been explored by \cite{8007057,jingjing,KKLMEC} for task offloading in multi-cell MEC systems and by \cite{li2019wireless} for data shuffling in wireless MapReduce systems, all with the goal of mitigating the multi-cell interference and hence boost the communication efficiency. Cooperative transmission has also been explored in the context of multi-cell caching systems in \cite{7857805,sengupta2017fog,FanXu,8624603,TaoCache} to accelerate content delivery by caching overlapped contents at different ENs.

\textbf{Overview and main contributions:} In the MEC system of Fig. \ref{sysModel}, investing more time for uplink communication allows the same user-generated input vectors to be received by more ENs, which enhances computation redundancy. The computation redundancy  generally introduces a heavier computation load, which can in turn increase the robustness against straggling servers via coded computing and mitigate multi-cell downlink interference via cooperative transmission. Based on these observations, this paper aims to establish the optimal trade-off between computing and download latencies at any given upload latency.  We focus on the high signal-to-noise ratio (SNR) regime in order to highlight the role of interference management as enabled by computation redundancy.

The most related prior works, as reviewed above, are \cite{jingjing} and \cite{KKLMEC}. The work \cite{jingjing} proposes a computing and downloading strategy by making the simplified assumption that the upload time is unconstrained so that the input vectors from all users are available at all ENs. The work \cite{KKLMEC} characterizes the trade-off between upload and download latencies by assuming that the computation time at each EN is deterministic (in contrast to random) so that coded computing is not needed. Moreover, reference \cite{KKLMEC} adopts a general task model, rather than matrix multiplication as studied in this work. In contrast to \cite{jingjing} and \cite{KKLMEC}, in this paper, we study the joint design of task assignment, input upload, edge computing, and output download, and we analyze the performance trade-offs among upload, computing and download latencies.

In summary, this paper studies the communication (in both uplink and downlink) and computation tradeoff in multi-user multi-server MEC networks by enabling the use of coded computing and cooperative transmission. The main contributions are as follows:
\begin{itemize}
\item We propose a new task offloading strategy that integrates coded computing based on a cascade of MDS and repetition codes\cite{improvedJingjing} with cooperative transmission at the ENs for interference management  \cite{FanXu}. By uploading the same input vectors of a user to multiple ENs, the policy creates data redundancy at the ENs that is leveraged to reduce the computation time by coded computing and the download time via transmission cooperation. Moreover, by waiting for more non-straggling ENs to finish their tasks, the proposed policy enhances the downlink transmission cooperation opportunities, and hence it allows the computation time to be traded for download time.
\item We derive achievable upload-compute-download latency triplets, as well as the end-to-end execution time, and characterize the trade-off region between computing and download latencies at any given upload latency. The analysis of upload and download latencies relies respectively on the analysis of degrees of freedom (DoF) for the effective X-multicast channel formed during uplink transmission and for the cooperative X channels obtained during downlink transmission\cite{FanXu}.
\item Furthermore, we provide a converse result that demonstrates the optimality of the achievable upload latency for fixed computation and download latencies, as well as constant multiplicative gaps to their  respective lower bounds for computation and download latencies at a large upload latency. The proof is based on genie-aided arguments and on a generalization of the arguments in
    \cite{sengupta2017fog}. The end-to-end execution time is also proved to be order-optimal for a large upload latency.
\item Through numerical examples, we show that, as compared to baseline schemes, the proposed policy can  reduce the overall end-to-end execution time. 
    We also show that, when the downlink transmission is the major bottleneck of the offloading process, the proposed cascaded MDS-repetition coding scheme reduces to repetition coding with no loss of optimality; while, when the bottleneck comes from the uplink transmission or edge computing, MDS coding is required to mitigate the effect of straggling ENs.
\end{itemize}

The rest of the paper is organized as follows. Section \ref{problemformu} presents the problem formulation and definitions. Main results including communication-computation latency trade-offs are presented in Section \ref{mainresults}. The proposed scheme is detailed in Section \ref{AchievableScheme}. Section \ref{simulation} provides numerical examples.  Conclusions are drawn in Section \ref{conclusion}.  The converse proof is available in Appendix.

Notations: $\mathcal{K}$ denotes the set of indexes $\{1,2,\cdots\!,K\}$. $[a\!:\!b]$ denotes the set of integers $\{a\!+\!1,a\!+\!2,\dots,b\}$. $[a]$ denotes the set of integers $[1\!:\!a]$. $(\cdot)^{T}$ denotes the transpose.  $(x)^+$ denotes $\max\{x,0\}$. $(X_i)_{i=a}^{b}$ denotes the vector $(X_a,X_{a+1},\cdots\!,X_b)^{T}$. $\{x_i\!:i\!\in\![a\!:\!b]\}$ or $\{x_i\}_{i=a}^{b}$ denotes the set $\{x_a, x_{a+1},\cdots\!,x_b\}$. $\{x_k\}_{q:K}$ denotes the $q$-th smallest element of set $\{x_k\!: k\!\in\![K]\}$. $\mathbb{F}^{m\times n}_{2^B}$ denotes the set of all matrices of dimension $m\times n$ with entries in the finite field $\mathbb{F}_{2^B}$\cite{Freedman1986Introduction}. 
\section{Problem Formulation} \label{problemformu}
\subsection{MEC Network Model}
As shown in Fig. \ref{sysModel}, we consider a multi-cell MEC network consisting of $K$ single-antenna ENs communicating with $M$ single-antenna users via a shared wireless channel. Denote by $\mathcal{K}\!=\!\{1,2,\ldots,K\}$ the set of ENs and by $\mathcal{M}\!=\!\{1,2,\ldots,M\}$ the set of users. Each EN is equipped with an edge server. 
 The mobile users offload their
computing tasks to the ENs through the uplink channel (from users to ENs) and then download
the computation results back via the downlink channel (from ENs to users). Let $h_{ki}^{\textnormal{u}}$ denote the uplink channel fading from user $i\!\in\!\mathcal{M}$ to EN $k\!\in\!\mathcal{K}$, and $h_{ik}^{\textnormal{d}}$ denote the downlink channel fading from EN $k\!\in\! \mathcal{K}$ to user $i\!\in\!\mathcal{M}$, both of which are independent and identically distributed (i.i.d.) for all pairs $(i,k)$ according to some continuous distribution.  
A central scheduling unit (CSU) is connected to all nodes via backhaul links to collect the uplink channel state information (CSI) $\mathbf{H}^\textnormal{u}\!\triangleq\!\{h_{ki}^{\textnormal{u}}\!:k\!\in\!\mathcal{K},i\!\in\!\mathcal{M}\}$ and downlink CSI $\mathbf{H}^\textnormal{d}\!\triangleq\!\{h_{ik}^{\textnormal{d}}\!:i\!\in\!\mathcal{M},k\!\in\!\mathcal{K}\}$ estimated at these nodes. It utilizes the collected global CSI to design the transmit or receive beamforming coefficients for the symbols transmitted or received at the nodes. We assume perfect CSI for uplink and downlink channels at the CSU, and we refer to \cite{AliCSI,JafarCSI,XCSI} for analysis of the impact of imperfect CSI in the high-SNR regime.

We consider that each user has a matrix multiplication task to compute. Matrix multiplication is a building block of many machine learning and data analytic problems, e.g., linear inference tasks including collaborative filtering for recommendation systems\cite{5197422}.
Specifically, we assume that each user $i$ has an input matrix with $N$ input vectors\footnote{The same number of input vectors for different users is assumed for analytical tractability. When this is not the case, a simple, generally suboptimal, solution is to add extra inputs (e.g., zero vectors) to each user to make their inputs equal in number.} $\mathbf{u}_{i,j}\!\in\!\mathbb{F}^{n\times1}_{2^B}$, $j\!\in\![N]$, and it wishes to compute the output matrix with $N$ output vectors $\mathbf{v}_{i,j}\!\in\!\mathbb{F}^{m\times 1}_{2^B}$, $j\!\in\![N]$, where
\begin{equation}\label{output}
\mathbf{v}_{i,j}\!=\! \mathbf{A}{\mathbf{u}_{i,j}},~\text{for}~j\!\in\![N],
\end{equation}
and $\mathbf{A}\! \in\!\mathbb{F}^{m\times n}_{2^B}$ is a data matrix available at the network end, and $B$ is the size (in bits) of each element. The matrix $\mathbf{A}$ is partially stored across the ENs, which conduct the product operations in a distributed manner. To this end, each EN $k$ has a storage capacity of $\mu m nB$ bits, and it can hence store a fraction $\mu\!\in\![\frac{1}{K},1]$ of the rows of matrix $\mathbf{A}$. Specifically, during an offline storage phase, an encoding matrix $\mathbf{E}_k\!\in\!\mathbb{F}^{\mu m\times m}_{2^B}$ is used to generate a coded matrix $\mathbf{A}_k\!=\!\mathbf{E}_k\mathbf{A}$, which is then stored at EN $k$, as in \cite{songzestraggler,improvedJingjing}.
\subsection{Task Offloading Procedure}
The task offloading procedure proceeds through task assignment, input uploading, edge computing, and output downloading.
\subsubsection{Task Assignment} A task assignment scheme is defined through the following sets \begin{equation}
\{\mathcal{U}_{i,\mathcal{K}^{'}}\!:i\!\in\!\mathcal{M},\mathcal{K}^{'}\!\subseteq\! \mathcal{K}\},
\end{equation}
where $\mathcal{U}_{i,\mathcal{K}^{'}}\!\subseteq\!\{\mathbf{u}_{i,j}\}^N_{j=1}$ denotes the subset of input vectors from user $i$ that are assigned only to the subset of ENs $\mathcal{K}^{'}$ for computation. We impose the condition $\bigcup_{\mathcal{K}^{'}\!\subseteq \mathcal{K}} \mathcal{U}_{i,\mathcal{K}^{'}}\!=\!\{\mathbf{u}_{i,j}\}^N_{j=1}$ for $i\!\in\!\mathcal{M}$ to guarantee that all input vectors are computed. Furthermore, by definition, we have the relation $\mathcal{U}_{i,\mathcal{K}^{'}}\!\bigcap \mathcal{U}_{i,\mathcal{K}^{''}}\!=\!\varnothing$ for $\mathcal{K}^{'}\!\neq\!\mathcal{K}^{''}$ so that these subsets are not overlapped. The subset of input vectors from all users assigned to each EN $k$ is hence given as $\mathcal{U}_k\!=\!\bigcup_{{i}\in\mathcal{M},\,\mathcal{K}^{'}\!\subseteq \mathcal{K}:\,k\in\mathcal{K}^{'}} \mathcal{U}_{i,\mathcal{K}^{'}}$.
\begin{definition}\label{defenition1}
(Repetition Order) For a given task assignment scheme $\{\mathcal{U}_{i,\mathcal{K}^{'}}\}_{i\in\mathcal{M},\mathcal{K}^{'}\subseteq\mathcal{K}}$, the repetition order $r$, with $1\!\le\!r \!\le \!K$, is defined as average input data redundancy, i.e., the total number of
input vectors assigned to the $K$ ENs (counting repetitions) divided by
the total number of input vectors of the $M$ users, i.e.,
\begin{equation}
r\!\triangleq\!\frac{\sum_{k\in\mathcal{K}}|\mathcal{U}_{k}|}{MN}.
\end{equation}
\end{definition}
The repetition order indicates the average number of ENs that are assigned the same input vector, which has been adopted in \cite{KKLMEC} as a measure of the degrees of computation replication.
The above task assignment $\{\mathcal{U}_{i,\mathcal{K}^{'}}\}$ is realized through the following input uploading phase.
\subsubsection{Input Uploading}
At run time, each user $i$ maps its input vectors $\{\mathbf{u}_{i,j}\}^N_{j=1}$
into a codeword $\mathbf{X}_i^{\textnormal{u}}\!\triangleq\!\left(X_i^{\textnormal{u}}(t)\right)^{T^\textnormal{u}}_{t=1}$ of length $T^\textnormal{u}$ symbols under the power constraint $(T^\textnormal{u})^{-1}\mathbb{E}\big[||\mathbf{X}_i^{\textnormal{u}}||^2\big]\!\le\! P^\textnormal{u}$. Note that $X_i^{\textnormal{u}}(t)\!\in\! \mathbb{C}$ is the symbol transmitted at time $t\!\in\![T^\textnormal{u}]$. At each EN $k\!\in\!\mathcal{K}$, the received signal $Y_k^{\textnormal{u}}(t) \!\in\! \mathbb{C}$ at time $t\!\in\![T^\textnormal{u}]$ can be expressed as
\begin{equation}
Y_k^{\textnormal{u}}(t)\!=\!\sum_{i \in \mathcal{M}}h_{ki}^{\textnormal{u}}(t)X_i^{\textnormal{u}}(t)\!+\!Z_{k}^{\textnormal{u}}(t),
\end{equation}
where $Z_{k}^{\textnormal{u}}(t)\!\sim\!\mathcal{CN}(0,1)$ denotes the noise at EN $k$. Each EN $k$ decodes the sequence $\left(Y_k^{\textnormal{u}}(t)\right)^{T^\textnormal{u}}_{t=1}$ into an estimate $\{\widehat{\mathbf{u}}_{i,j}\}$ of the assigned input vectors $\{\mathbf{u}_{i,j}\!:\mathbf{u}_{i,j}\!\in\!\mathcal{U}_k\}$.
\subsubsection{Edge Computing}
After the uploading phase is completed, each EN $k$ computes the products of the assigned estimated input vectors in set $\mathcal{U}_{k}$ with its stored coded model $\mathbf{A}_k$. 
The computation time for EN $k$ to complete the computation of the corresponding $\mu m |\mathcal{U}_{k}|$ row-vector products\footnote{The row-vector product indicates the product of a row vector of matrix $\mathbf{A}$ with a column vector $\mathbf{u}_{i,j}$.} is modeled as
\begin{equation} \label{computingtime}
T_k^{\textnormal{c}} = \mu m |\mathcal{U}_{k}| \omega_k, ~\text{for}~k\in\mathcal{K},
\end{equation}
where the random variable $\omega_k$ represents the time needed by EN $k$ to compute a row-vector product, and it is modelled as an exponential distribution with mean $1/\eta$ (see, e.g., \cite{Speeding_Up,songzestraggler,8006960,8502151}). $T_k^{\textnormal{c}}$ is thus a scaled exponential distribution with mean $ \mu m |\mathcal{U}_{k}|/\eta$.
The MEC network waits until the fastest $q$ ENs, denoted as subset $\mathcal{K}_q\!\subseteq\!\mathcal{K}$, have finished their tasks before returning the results back to users in the downlink. The cardinality $|\mathcal{K}_q|\!=\!q$ is referred to as \emph{the recovery order}.
The rest of $K\!-\!q$ ENs are known as \emph{stragglers}. The resulting (random) duration of the edge computing phase is hence equal to the maximum computation time of the $q$ fastest ENs, i.e.,
$T^\textnormal{c}=\max_{k\in\mathcal{K}_q} T_k^{\textnormal{c}}$.
\subsubsection{Output Downloading}
At the end of the edge computing phase, each EN $k\!\in\!\mathcal{K}_q$ obtains the coded outputs $\mathcal{V}_k\!\triangleq\!\left\{\mathbf{v}_{i,j,k}\!=\!\mathbf{A}_k\widehat{\mathbf{u}}_{i,j}\!:\mathbf{u}_{i,j}\!\in\!\mathcal{U}_{k}\right\}$. 
Every EN $k$ in $\mathcal{K}_q$ then maps $\mathcal{V}_k$ 
into a length-$T^\textnormal{d}$ codeword $\mathbf{X}_k^{\textnormal{d}}\!\triangleq\!\left(X_{k}^{\textnormal{d}}(t)\right)^{T^\textnormal{d}}_{t=1}$ with an average power constraint $(T^\textnormal{d})^{-1}\mathbb{E}\big[||\mathbf{X}_k^{\textnormal{d}}||^2\big]\!\le\! P^\textnormal{d}$. For each user $i\!\in\!\mathcal{M}$, its received signal $Y_i^{\textnormal{d}}(t) \!\in\! \mathbb{C}$ at time $t\!\in\![T^\textnormal{d}]$ is given by
\begin{equation}
Y_i^{\textnormal{d}}(t)\!=\!\sum_{k \in \mathcal{K}_q}h_{ik}^{\textnormal{d}}(t)X_{k}^{\textnormal{d}}(t)\!+\!Z_{i}^{\textnormal{d}}(t),
\end{equation}
where $Z_{i}^{\textnormal{d}}(t)\!\sim\!\mathcal{CN}(0,1)$ is the noise at user $i$. Each user $i$ decodes the sequence $(Y_i^{\textnormal{d}}(t))^{T^\textnormal{d}}_{t=1}$ to obtain an estimate $\{\widehat{\mathbf{v}}_{i,j,k}\}_{j\in[N],k\in\mathcal{K}_q}$ of the coded outputs, from which it obtains an estimate  $\{\widehat{\mathbf{v}}_{i,j}\}_{j\in[N]}$ of its desired outputs. 
This is possible if the estimated coded outputs $\{\widehat{\mathbf{v}}_{i,j,k}\}_{j\in[N],k\in\mathcal{K}_q}$ contain enough information to guarantee the condition  $H(\{\mathbf{v}_{i,j}\}_{j\in[N]}|\{\widehat{\mathbf{v}}_{i,j,k}\}_{j\in[N],k\in\mathcal{K}_q})\!=\!0$. The overall error probability is given as $\mathrm{P}_\text{e}\!\triangleq\!\mathbb{P}\big(\bigcup^{M~~N}_{i=1,j=1} \left\{\widehat{\mathbf{v}}_{i,j}\!\neq\!\mathbf{v}_{i,j}\right\}\!\big)$.
A task offloading policy is said to be feasible when the error probability $\mathrm{P}_\text{e}\!\to\!0$ as $B\!\to\!\infty$.
\subsection{Performance Metric}
The performance of the considered MEC network is characterized by the  latency triplet accounting for task uploading, computing, and output downloading, which we measure in the high-SNR regime as defined below.
\begin{definition}\label{defenition2}
The normalized uploading time (NULT), normalized computation time (NCT), and normalized downloading time (NDLT) achieved by a feasible policy with repetition order $r$ and recovery order $q$ are defined, respectively, as
\begin{align}
\tau^\textnormal{u}(r) &\triangleq \lim_{P^\textnormal{u}\to\infty} \lim_{B\to\infty} \frac{\mathbb{E}_{\mathbf{H}^\textnormal{u}}[T^\textnormal{u}]}{ NnB/\log P^\textnormal{u}}, \label{NULTtau}\\
\tau^\textnormal{c} (r,q) &\triangleq \lim_{m\to\infty}\frac{\mathbb{E}_{\boldsymbol{\omega}}\left[T^\textnormal{c}\right]}{ Nm/\eta},\label{NCTtau}\\
\tau^\textnormal{d}(r,q) &\triangleq \lim_{P^\textnormal{d}\to\infty}\lim_{m\to\infty}\lim_{B\to\infty} \frac{\mathbb{E}_{\mathbf{H}^\textnormal{d}}[T^\textnormal{d}]}{NmB/\log P^\textnormal{d}}. \label{NDLTtau}
\end{align}
\end{definition}
The definitions (\ref{NULTtau}) and (\ref{NDLTtau}) have been also adopted in \cite{KKLMEC}, and follow the approach introduced in \cite{sengupta2017fog} by normalizing the delivery times to those of reference interference-free systems (with high-SNR rates $\log P^\textnormal{u}$ and $\log P^\textnormal{d}$, respectively). Similarly, the computation time in definition (\ref{NCTtau}) is normalized by the average time needed to compute over all the input vectors of a user. To avoid rounding complications, in  definition (\ref{NCTtau}) and  (\ref{NDLTtau}), we let the output dimension $m$ grow to infinity.
\begin{definition} \label{def3}
Given the definition of achievable NULT-NCT-NDLT triplet $\left(\tau^{\textnormal{u}}(r),\tau^{\textnormal{c}}(r,q),\tau^{\textnormal{d}}(r,q)\right)$ with repetition order $r$ and recovery order $q$ as in Definition \ref{defenition2}, the optimal compute-download latency region for a given NULT $\tau^{\textnormal{u}}$ is defined as the union of all NCT-NDLT pairs $(\tau^\textnormal{c}, \tau^{\textnormal{d}})$ that satisfy $\tau^{\textnormal{c}}\!\ge\! \tau^{\textnormal{c}}(r,q)$ and $\tau^{\textnormal{d}}\!\ge\!\tau^{\textnormal{d}}(r,q)$ for some $(r,q)$ while the corresponding NULT $\tau^\textnormal{u}(r)$ is no larger than $\tau^\textnormal{u}$, i.e.,
\begin{align}\label{region}
\!\!\!\mathscr{T}^{*}(\tau^\textnormal{u})\!\triangleq\!\big\{\!(\tau^\textnormal{c}, \tau^\textnormal{d})\!: \left(\tau^\textnormal{u}(r),\tau^\textnormal{c}(r,q),\tau^\textnormal{d}(r,q)\right)~&\text{is achievable for}~\text{some $(r,q)$ and}\,\,\tau^\textnormal{u}(r)\!\le\!\tau^\textnormal{u},\nonumber\\
\!&~~~~~~~~~~~~~~\tau^\textnormal{c}(r,q)\!\le\!\tau^\textnormal{c},~\text{and}~\tau^\textnormal{d}(r,q)\!\le\!\tau^\textnormal{d}\big\}.
\end{align}
\end{definition}
\begin{definition}\label{totaltimeremark}(End-to-end execution time) For a given pair $(r,q)$, based on the defined communication and computation latency triplet, the end-to-end execution time is defined as the weighted sum of the NULT, NCT, and NDLT as
\begin{equation}
\tau(r,q) = \tau^\textnormal{u}(r) + \delta_c\tau^\textnormal{c} (r,q) + \delta_d\tau^\textnormal{d}(r,q). \label{totaltime}
\end{equation}
In (\ref{totaltime}), $\delta_c\!=\!\frac{Nm/\eta}{NnB/\log P^\textnormal{u}}$ represents the ratio between the reference time needed to compute over all the input vectors of a user and the reference time needed to upload all the input vectors of a user, while $\delta_d\!=\!\frac{NmB/\log P^\textnormal{d}}{NnB/\log P^\textnormal{u}}$ is ratio between the reference time needed to download all the output vectors of a user and the mentioned upload reference time. \end{definition}
\begin{remark} \label{remarkconvex}
(Convexity of compute-download latency region.) For an input data assignment policy $\{\mathcal{U}_{i,\mathcal{K}^{'}}\}_{i\in\mathcal{M},\mathcal{K}^{'}\subseteq\mathcal{K}}$ with repetition order $r$, fix an input uploading strategy achieving an NULT of $\tau^\textnormal{u}$. Consider now two policies $\pi_1$ and $\pi_2$ that differ may in their computing and download phases, and achieve two NCT-NDLT pairs $\left(\tau_{1}^{\textnormal{c}}, \tau_{1}^{\textnormal{d}}\right)$ and $\left(\tau_{2}^{\textnormal{c}}, \tau_{2}^{\textnormal{d}}\right)$, respectively. For any ratio $\lambda\!\in\![0,1]$, it can be seen that there exists a policy that achieves the NCT-NDLT pair $\lambda\left(\tau_{1}^{\textnormal{c}}, \tau_{1}^{\textnormal{d}}\right) \!+\!(1\!-\!\lambda)\left(\tau_{2}^{\textnormal{c}}, \tau_{2}^{\textnormal{d}}\right)$ for the same NULT $\tau^\textnormal{u}$. To this end, assuming $m$ is sufficiently large, matrix $\mathbf{A}$, correspondingly, all output vectors in (\ref{output}) are split horizontally so that 
$N\lambda m$ and $N(1\!-\!\lambda)m$ outputs can be processed by using policies $\pi_1$ and $\pi_2$, respectively. 
By the linearity of the NCT in (\ref{NCTtau}) and NDLT in (\ref{NDLTtau}) with respect to the output size, the claimed pair of NCT and NDLT is achieved. Thus, the region in (\ref{region}) is convex. Similar arguments were also used in \cite[Lemma 1]{sengupta2017fog}.
\end{remark}
Based on the remark above, the region $\mathscr{T}^{*}(\tau^\textnormal{u})$ in (\ref{region}) is convex thanks to the time- and memory-sharing arguments, while it can be proved that the same is not true for the region of achievable triplets $(\tau^\textnormal{u}, \tau^\textnormal{c}, \tau^\textnormal{d})$. Region $\mathscr{T}^{*}(\tau^\textnormal{u})$ will be adopted to capture the trade-offs between computation and download latencies for a fixed upload latency. Our general goals are to characterize the minimum communication-computation latency triplet, the optimal tradeoff region between computing and download latencies, as well as the minimum end-to-end execution time.
\begin{figure}[t]
\centering
\includegraphics[width=3.87in, height=2.3in]{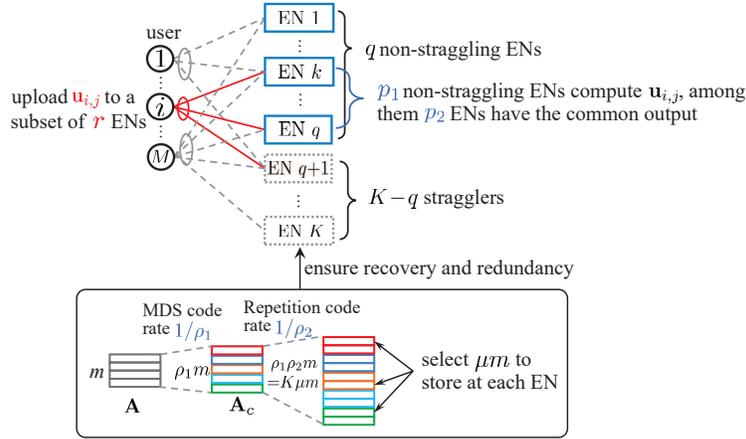}
\vspace{-2mm}
\caption{(Bottom) Hybrid MDS-Repetition coding for matrix $\mathbf{A}$; (Top) Input uploading and edge computing.}
\vspace{-8mm}
\label{mds_repetition}
\end{figure}
\vspace{-4mm}
\section{Main Results}\label{mainresults}
In this section, we introduce a novel task offloading scheme based on the joint design of task assignment, two-way communication, and cascaded coded computing. Then, we study the communication-computation latency triplet and the end-to-end execution time achieved by this scheme. We derive inner and outer bounds on the compute-download latency
region, and then discuss some consequences of the main results in terms of the tradeoffs among
upload, computation, and download latencies. Furthermore, we specialize the main results to a
number of simpler set-ups in order to illustrate the connections with existing works.
\vspace{-2mm}
\subsection{Key Ideas} \label{keyidea}
We start by outlining the main ideas that underpin the proposed scheme. In task assignment, we choose a repetition-recovery order pair $(r,q)$ from a feasible set $\mathcal{R}$ of values. We demonstrate that for any $(r,q)\!\in\!\mathcal{R}$, it is possible to recover all outputs through a suitable design of the system. For any such pair $(r,q)$, as shown in Fig. \ref{mds_repetition}-(bottom), matrix $\mathbf{A}$ is encoded by a cascade of an MDS code of rate $1/\rho_1$ and a repetition code of rate $1/\rho_2$. Of the encoded rows, $\mu m$ different rows are stored at each EN, with each MDS encoded row replicated at $\rho_2$ distinct ENs. As we will prove, the MDS code can alleviate the impact of stragglers on the computation latency by decreasing the admissible values for the number $q$ of non-straggling ENs (see also \cite{Speeding_Up,songzestraggler,improvedJingjing}); while repetition coding can reduce the download latency by enabling cooperative transmission among multiple ENs computing the same outputs\cite{jingjing,KKLMEC}.

In the input upload phase, each user divides its $N$ input vectors into $\binom{K}{r}$ subsets $\{\mathcal{U}_{i,\mathcal{K}^{'}}\}$, with each subset uploaded to all the $r$ ENs in subset $\mathcal{K}^{'}$ for computation. By using interference alignment (IA), at each EN, a total of $M\binom{K-1}{r-1}$ desired subsets of inputs can be successfully decoded with the other $M\binom{K-1}{r}$ undesired subsets of interfering signals being aligned together. Then, as shown in Fig. \ref{mds_repetition}-(top), in the computing phase, each input vector of any user is computed by a subset of $p_1$ non-straggling ENs with $p_1$ being at least $r-(\!K\!-q)$ and at most $\min\{r,q\}$. Therefore, since each encoded row of $\mathbf{A}$ is replicated at a subset of $\rho_2$ ENs, after computation, each MDS-encoded row-vector product result for a user will be replicated at a subset of $p_2$ non-straggling ENs, with $p_2$ being at least $\max\{\rho_2\!-\! K\!+\!p_1,1\}$ and at most $\min\{p_1,\rho_2\}$.

In the output download phase, each subset of $p_2$ ENs computing the same coded outputs can first use zero-forcing (ZF) precoding to null the interfering signal caused by common outputs at a subset of $p_2-\!1$ undesired users. When the number of undesired users does not exceed $p_2-\!1$, i.e., when $M\!-\!1\!\le\!p_2-\!1$, by ZF precoding, each user only receives its desired outputs with all undesired outputs being cancelled out. When this condition is violated, i.e., when $M\!>\!p_2$, after ZF precoding, each output still causes interferences to $M\!-\!p_2$ undesired users. As detailed in \cite{FanXu}, IA can be applied in cascade to the ZF precoders in order to mitigate the impact of these interfering signals. 
\vspace{-1mm}
\subsection{Bounds} \vspace{-1mm}
The scheme summarized above and detailed in Sec. \ref{AchievableScheme} achieves the following latency region.
\begin{theorem} \label{achievableresults}
(Inner bound). For the described MEC network with $M$ users and $K$ ENs, 
each with storage capacity $\mu\!\in\![\frac{1}{K},1]$, the following communication-computation latency triplet $\big(\tau_{\textnormal{a}}^{\textnormal{u}}(r),$ $\tau_{\textnormal{a}}^{\textnormal{c}}(r,q),\tau_{\textnormal{a}}^{\textnormal{d}}(r,q)\big)$ is achievable
\begin{align}
\!\tau_{\textnormal{a}}^{\textnormal{u}}(r)&\!=\! \frac{(M\!-\!1)r \!+\! K}{K}, \label{achieveUpoad} \\
\!\tau_{\textnormal{a}}^{\textnormal{c}}(r,q)&\!=\!\frac{Mr\mu(H_K\!-\!H_{K-q})}{K}\!, \label{achievecompute}\\
\!\tau_{\textnormal{a}}^{\textnormal{d}}(r,q)&\!=\!\sum\limits^{\min\{r,q\}}_{p_1\!=r-\!K\!+q}\!\!\!\!\!B_{p_1}\!\!\left(\sum\limits^{l_\textnormal{max}}_{p_2=l_{p_1}}\!\!\frac{B_{p_2}}{d_{p_1,M,p_2}^{\textnormal{d}}}\!+\!
\!\frac{B_{l_{p_1}\!-1}}{d_{p_1,M,l_{p_1}\!-1}^{\textnormal{d}}}\!\right)\!, \label{achievedownload}
\end{align}
for any repetition order $r$ and recovery order $q$ in the set
\begin{equation} \label{regionR}
\mathcal{R}\!\triangleq\!\big\{(r,q)\!:r\!\in\![K],q\!\in\![K],~\text{and}~ (r\!-\!K\!+\!q)\mu\!\ge\!1\big\},
\end{equation}
where $H_K\!=\!\sum^K_{k=1}1/k$, $H_0\!=\!0$, $B_{p_1}\!=\!\binom{q}{p_1}\binom{K-q}{r-p_1}/\binom{K}{r}$, $B_{p_2}\!=\!\binom{p_1}{p_2}\binom{K-p_1}{\rho_2-p_2}\rho_1/\binom{K}{\rho_2}$, $B_{l_{p_1}\!-1}\!=\!1\!-\!\!\sum^{l_\textnormal{max}}_{p_2=l_{p_1}}\!\!B_{p_2}$, and $d_{p_1,M,p_2}^{\textnormal{d}}$ is given by
\begin{equation}\label{dofdd}
d_{p_1,M,p_2}^{\textnormal{d}}=\left\{
\begin{aligned}
&1, &p_2\!\ge\!M~~~~\\
&\frac{\binom{p_1}{M-1}(M\!-\!1)}{\binom{p_1}{M-1}(M\!-\!1)+1},   &~~p_2\!=\!M\!-\!1\\
&\max\left\{d', \frac{p_2}{M}\right\},          &~~p_2\!\le\!M\!-\!2
\end{aligned}
~,\right.
\end{equation}
with $d'\!\triangleq\!\max_{1\le t\le p_2}\frac{p_1-t+1}{M+p_1-2t+1}$;
\begin{align} \label{theoremcoderate}
\!\!\rho_2\!=\!\inf\!\bigg\{\rho\!:\!\!\binom{K}{\rho}\!-\!\binom{2K\!-\!r\!-\!q}{\rho}\!\ge\! \frac{1}{\rho_1}\binom{K}{\rho},\rho_1\rho\!=\!K\mu,
\rho_1\!\in\!\big\{1,\frac{K\mu}{K\mu\!-\!1},\frac{K\mu}{K\mu\!-\!2}\cdots,K\mu\big\}\!\bigg\};\!\!
\end{align}
and
\begin{align}\label{replicationmin}
\!\!\!l_{p_1} \!\!=\! \inf\!\bigg\{l\!:\!\!\sum\nolimits^{l_\textnormal{max}}_{p_2=l}\!\!B_{p_2}m\!\le\!m, l\!\in\![l_\textnormal{min}\!:\!l_\textnormal{max}], l_\textnormal{max}\!=\!\min\{p_1,\rho_2\}, 
l_\textnormal{min}\!=\!\max\{\rho_2\!-\! K\!+\!p_1,1\} \bigg\}.\!
\end{align}
Therefore, for an NULT $\tau^\textnormal{u}\!=\!\tau_{\textnormal{a}}^{\textnormal{u}}(r)$ given in (\ref{achieveUpoad}) for some $r$, an inner bound $\mathscr{T}_{in}(\tau^\textnormal{u})$ on the compute-download latency region is given by the convex hull of the set $\big\{\!\left(\tau_{\textnormal{a}}^{\textnormal{c}}(r,q),\tau_{\textnormal{a}}^{\textnormal{d}}(r,q)\right)\!:\!q\!\in\!\big[\lceil\frac{1}{\mu}\rceil+K\!-r\!:\!K\big]\big\}$.
\end{theorem}
\begin{proof}The proof of Theorem \ref{achievableresults} is given in Section \ref{AchievableScheme}. \end{proof}
By Theorem \ref{achievableresults}, an achievable end-to-end execution time is given as follows.
\begin{corollary}
An achievable end-to-end execution time for $(r,q)\!\in\!\mathcal{R}$
is given as $\tau_{\textnormal{a}}(r,q)\!=\!\tau_{\textnormal{a}}^{\textnormal{u}}(r)\!+\delta_c\tau_{\textnormal{a}}^{\textnormal{c}}(r,q)+\delta_d\tau_{\textnormal{a}}^{\textnormal{d}}(r,q)$, where $\tau_{\textnormal{a}}^{\textnormal{u}}(r)$, $\tau_{\textnormal{a}}^{\textnormal{c}}(r,q)$, and $\tau_{\textnormal{a}}^{\textnormal{d}}(r,q)$ are given in (\ref{achieveUpoad}), (\ref{achievecompute}), and (\ref{achievedownload}), respectively.
\end{corollary}
We also have the following converse.
\begin{theorem} \label{lower_bound}
(Converse). For the same MEC network, the set of all admissible pairs $(r,q)$ is included in the set $\mathcal{R}$ in (\ref{regionR}). Furthermore, any feasible communication-computation latency triplet $\left(\tau^\textnormal{u}(r),\tau^\textnormal{c}(r,q),\tau^\textnormal{d}(r,q)\right)$ for pairs $(r,q)$ in $\mathcal{R}$ is lower bounded as
\begin{align}
\!\tau^\textnormal{u}(r)&\!\ge\! \tau_{\textnormal{a}}^{\textnormal{u}}(r),\!\label{lowerupload}\\
\!\tau^\textnormal{c}(r,q)&\!\ge\!\tau_l^{\textnormal{c}}(r,q)\!=\!\max\limits_{t\in[q]}\!\frac{(H_K\!-\! H_{K-q+t-1})(r\!-\!K\!+\!t)^{+}M\mu }{t}\!,\!\label{lowercompute} \\
\!\tau^\textnormal{d}(r,q)&\!\ge\!\tau_l^{\textnormal{d}}(r,q)\!=\!\max_{t\in\{1,\cdots,\min\{q,M\}\}}\!\frac{M\!-\!(M\!-\!t)(q\!-\!t)\frac{r}{K}\mu}{t}\!.\! \label{lowerdownload}
\end{align}
Therefore, for an NULT $\tau^\textnormal{u}\!=\!\tau_{\textnormal{a}}^{\textnormal{u}}(r)$ in (\ref{achieveUpoad}) for some $r$, an outer bound $\mathscr{T}_{out}(\tau^\textnormal{u})$ of the compute-download latency region is given by the convex hull of set
$\big\{\!\!\left(\tau_l^{\textnormal{c}}(r,q),\tau_l^{\textnormal{d}}(r,q)\right)\!\!:\!q\!\in\!\!\big[\lceil\!\frac{1}{\mu}\!\rceil\!+\!K\!-r\!:\!\!K\big]\!\big\}$.
\end{theorem}
\begin{proof}The proof of Theorem \ref{lower_bound} is available in Appendix.\end{proof} 
\begin{figure*}
\vspace{-5mm}
\centering
\subfigure[$\tau^\textnormal{u}\!=\!6.4$ ($r\!=\!6$)]
{\label{bR} 
 \includegraphics[width=3in, height=2.3in]{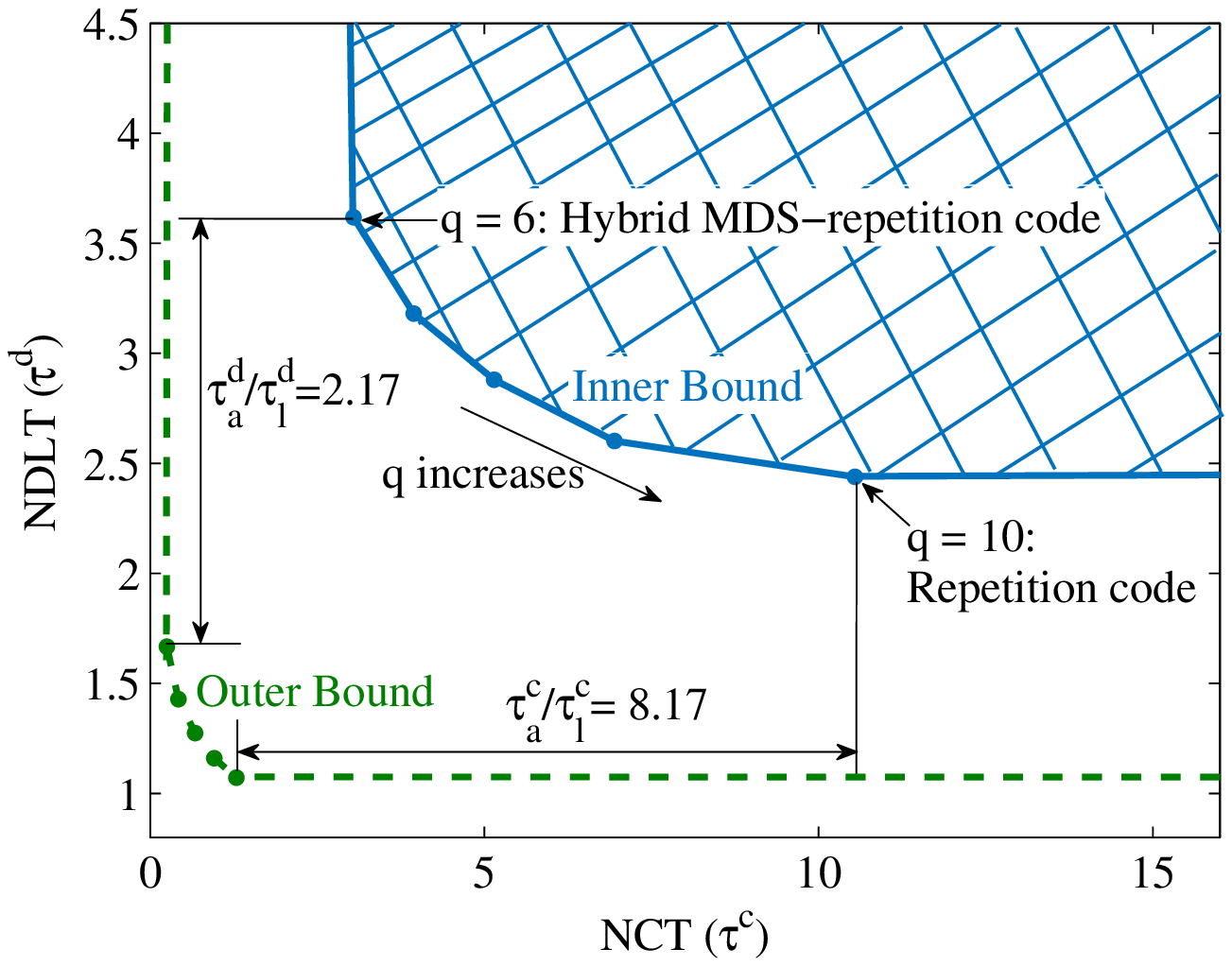}}
\hspace{4mm}
\subfigure[$\tau^\textnormal{u}\!=\!9.1$ ($r\!=\!9$)]
{\label{sR} 
 \includegraphics[width=3in, height=2.3in]{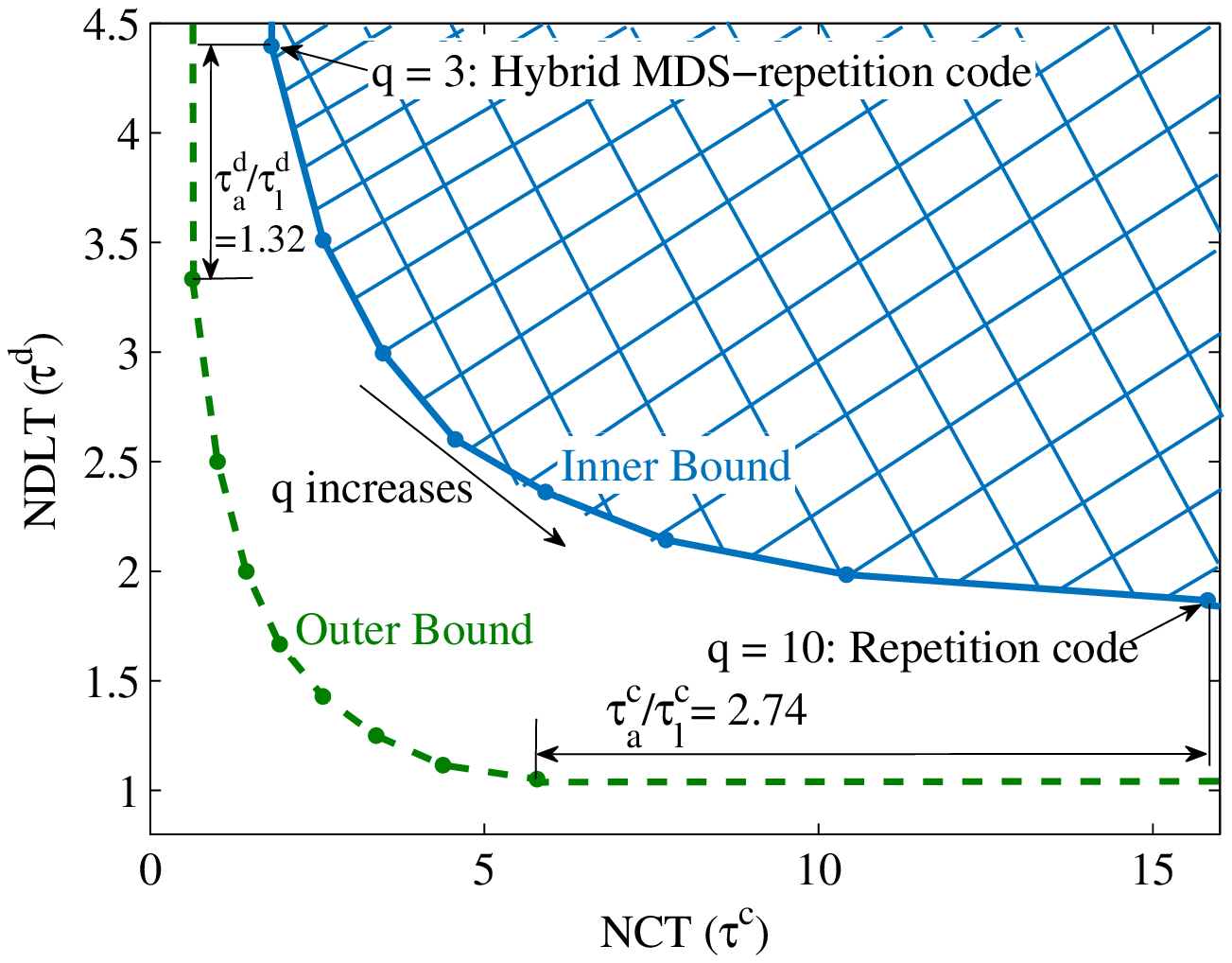}}
 \vspace{-2mm}
 \caption{Compute-download latency region bounds for $M\!=\!K\!=\!10$, $\mu\!=\!3/5$, $N=\binom{r}{10}$, and $m\!=\!\rho_2\binom{\rho_2}{10}$.}
\label{latencytriples} 
\vspace{-8mm}
\end{figure*}
Fig. \ref{latencytriples} plots the derived inner and outer bounds on the compute-download latency region $\mathscr{T}^{*}(\tau^\textnormal{u})$ for the case with $M\!\!=\!\!K\!\!=\!\!10$ and two different values of $\tau^\textnormal{u}$. For instance, in Fig. \ref{latencytriples}-(a), for a small NULT $\tau^\textnormal{u}\!=\!6.4$ at $r\!=\!6$, we have the achievable NCT-NDLT pairs $(\tau_{\textnormal{a}}^{\textnormal{c}},\tau_{\textnormal{a}}^{\textnormal{d}})\!=\!(3.04,3.62)$ at $q\!=\!6$ and $(\tau_{\textnormal{a}}^{\textnormal{c}},\tau_{\textnormal{a}}^{\textnormal{d}})\!=\!(10.54,2.44)$ at $q\!=\!10$; while, in Fig. \ref{latencytriples}-(b), for a large NULT $\tau^\textnormal{u}\!=\!9.1$ at $r\!=\!9$, we have two smaller latency pairs $(\tau_{\textnormal{a}}^{\textnormal{c}},\tau_{\textnormal{a}}^{\textnormal{d}})\!=\!(2.59,3.51)$ at $q\!=\!4$ and $(\tau_{\textnormal{a}}^{\textnormal{c}},\tau_{\textnormal{a}}^{\textnormal{d}})\!=\!(7.72,2.14)$ at $q\!=\!8$. First, we observe that for both cases,
\emph{as $q$ increases, the NDLT is reduced at the expense of an increasing NCT}: A larger $q$ enables more opportunities for transmission cooperation at the ENs during output downloading, while increasing, on average, the time required for $q$ ENs to complete their tasks. Furthermore, comparing Fig. \ref{latencytriples}-(a) with Fig. \ref{latencytriples}-(b), we also see that \emph{allowing for a longer upload time $\tau^\textnormal{u}$ increases the compute-download latency region}. This is because  when more information is uploaded to ENs over a larger latency $\tau^\textnormal{u}$, on the one hand, users can wait for fewer ENs to finish their computing tasks, reducing the NCT; and, on the other hand, the increased duplication of outputs also increases opportunities for transmission cooperation to reduce the NDLT.
\begin{corollary}
The minimum end-to-end execution time $\tau(r,q)$ for $(r,q)\!\in\!\mathcal{R}$
is lower bounded as $\tau(r,q)\!\ge\!\tau_l(r,q)\!=\!\tau_{\textnormal{a}}^{\textnormal{u}}(r)+\delta_c\tau_l^{\textnormal{c}}(r,q)+\delta_d\tau_l^{\textnormal{d}}(r,q)$, where $\tau_{\textnormal{a}}^{\textnormal{u}}(r)$, $\tau_l^{\textnormal{c}}(r,q)$, and $\tau_l^{\textnormal{d}}(r,q)$ are given in (\ref{achieveUpoad}), (\ref{lowercompute}), and (\ref{lowerdownload}), respectively.
\end{corollary}
\subsection{Optimality}
The following lemma characterizes the optimality of the proposed scheme.
\begin{lemma} \label{gapLemma}
(Optimality). For any triplet $\left(\tau^\textnormal{u}_\textnormal{a}(r),\tau^\textnormal{c}(r,q),\tau^\textnormal{d}(r,q)\right)$ with NULT $\tau^\textnormal{u}_\textnormal{a}(r)$ in (\ref{achieveUpoad}), it is not possible to reduce the achievable NULT $\tau^\textnormal{u}$ while still guaranteeing the feasibility of a triplet $(\tau^\textnormal{u},\tau^\textnormal{c}(r,q),\tau^\textnormal{d}(r,q))$. Furthermore, for a sufficiently large NULT $\tau^\textnormal{u}\!\ge\!\tau_{\textnormal{a}}^{\textnormal{u}}(K\!-\!n_1)$ and small recovery order $q\!\le\!K(1\!-\!1/n_2)\!+\!1$, with integers   $0\!\le\!n_1\!\!<\!q/2$ and $n_2\!\ge\!1$, the multiplicative gap between the achievable NCT in (\ref{achievecompute}) and its lower bound $\tau_l^{\textnormal{c}}$ in (\ref{lowercompute}) satisfies the inequality
\begin{equation}\label{gapnct}
\tau_{\textnormal{a}}^{\textnormal{c}}/\tau_l^{\textnormal{c}}\!\le\!(1+n_1) (1+n_2).
\end{equation}
Finally,  for a sufficiently large NULT $\tau^\textnormal{u}\!\ge\!\tau_{\textnormal{a}}^{\textnormal{u}}(K\!-\!n)$, with integer $n\!\ge\! 0$, the multiplicative gap between the achievable NDLT in (\ref{achievedownload}) and its lower bound $\tau_l^{\textnormal{d}}$ in (\ref{lowerdownload}) satisfies the inequality
\begin{equation} \label{gapndlt}
\tau_{\textnormal{a}}^{\textnormal{d}}/\tau_l^{\textnormal{d}}\!\le\!2(1\!+\!n\mu),
\end{equation}
and hence, if $\tau^\textnormal{u}\!\ge\!\tau_{\textnormal{a}}^{\textnormal{u}}(K)$, we have $\tau_{\textnormal{a}}^{\textnormal{d}}/\tau_l^{\textnormal{d}}\!\le\!2$. In the special case $\mu\!=\!1$, for  a sufficiently large NULT $\tau^\textnormal{u}\!\ge\!\tau_{\textnormal{a}}^{\textnormal{u}}(M\!+\!K\!-q)$, we have $\tau_{\textnormal{a}}^{\textnormal{d}}\!=\tau_l^{\textnormal{d}}\!=\!1$ that is optimal; for a smaller NULT $\tau_{\textnormal{a}}^{\textnormal{u}}(K\!-\!n)\!\le\!\tau^\textnormal{u}\!<\!\tau_{\textnormal{a}}^{\textnormal{u}}(M\!+\!K\!-\!q)$, with integer $\!q\!-\!M\!<\!n\!\le\!q\!-\!1$, we have $\tau_{\textnormal{a}}^{\textnormal{d}}/\tau_l^{\textnormal{d}}<\!n\!+\!1$.
\end{lemma}
\begin{proof}The proof of Lemma \ref{gapLemma} is given in Appendix.\end{proof}
The multiplicative gaps in Fig. \ref{latencytriples} are consistent with Lemma \ref{gapLemma}, since  $\tau_{\textnormal{a}}^{\textnormal{c}}/\tau_l^{\textnormal{c}}\!=\!2.74\!<\!22$ at $(r,q)\!=\!(9,10)$ (i.e., $n_1\!=\!1$ and $n_2\!=\!10$) and $\tau_{\textnormal{a}}^{\textnormal{d}}/\tau_l^{\textnormal{d}}\!=\!1.32\!<\!3.2$ at $(r,q)\!=\!(9,3)$ (i.e., $n\!=\!1$). Based on the inequality $\tau_{\textnormal{a}}/\tau_l\!=\!(\tau_{\textnormal{a}}^{\textnormal{u}}+\delta_c\tau_{\textnormal{a}}^{\textnormal{c}}+\delta_d\tau_{\textnormal{a}}^{\textnormal{d}})/(\tau_{\textnormal{a}}^{\textnormal{u}}+\delta_c\tau_l^{\textnormal{c}}+\delta_d\tau_l^{\textnormal{d}})\!\le\!\max\{\tau_{\textnormal{a}}^{\textnormal{c}}/\tau_l^{\textnormal{c}},\tau_{\textnormal{a}}^{\textnormal{d}}/\tau_l^{\textnormal{d}}\}$, the order-optimality of the end-to-end execution time is obtained as below.
\begin{corollary}
For a sufficiently large NULT $\tau^\textnormal{u}\!\ge\!\tau_{\textnormal{a}}^{\textnormal{u}}(K\!-\!n_1)$ and small recovery order $q\!\le\!K(1\!-\!1/n_2)\!+\!1$, with integers $0\!\le\!n_1\!\!<\!q/2$ and $n_2\!\ge\!1$, the multiplicative gap between the achievable end-to-end execution time $\tau_{\textnormal{a}}$ and its lower bound $\tau_l$ satisfies the inequality
\begin{equation}
\tau_{\textnormal{a}}/\tau_l\!\le\!\max\{(1+n_1) (1+n_2),2(1\!+\!n_1\mu)\}.
\end{equation}
\end{corollary}
\vspace{-4mm}
\subsection{Special Cases}
In the special case when $r\!=\!K$, hence ignoring limitations on the uplink transmission, the achievable NDLT (\ref{achievedownload}) reduces to $\tau_{\textnormal{a}}^{\textnormal{d}}(K,q)\!=\!M\sum^{l_\textnormal{max}}_{p_2=l_{q}}\!\!B_{p_2}/p_2\!+\!
\!B_{l_{p_1}\!-1}/(l_{p_1}\!-\!1)$ when using only ZF precoding in downlink, which is consistent with the normalized communication delay in \cite[Eq. (13)]{jingjing}. Furthermore, when setting $q\!=\!K$, hence ignoring stragglers{'} effects, and $\mu\!=\!1$, i.e., ignoring ENs{'} storage constraint, the achievable NDLT (\ref{achievedownload}) reduces to $\tau_{\textnormal{a}}^{\textnormal{d}}\!=\!M/\min\{K,M\}$, which is optimal and recovers the communication load in \cite[Remark 5]{8007057}, the NDT with cache-aided EN cooperation in \cite[Eq. (25)]{sengupta2017fog}, and the NDLT in \cite[Eq. (50)]{KKLMEC}.
\section{Achievable Scheme} \label{AchievableScheme}
In this section, we present the achievable scheme for any $\mu\!\in\!\!\{1/K,2/K,\cdots\!,1\}$\footnote{For general $u\!\in\![\frac{1}{K},1]$ satisfying $K\mu\!=\!\beta \lceil K\mu\rceil\!+\!(1\!-\!\beta)\lfloor K\mu\rfloor$, we can use memory- and time-sharing methods to achieve the linear combinations of the latency triplets achieved at integers $\lceil K\mu\rceil$ and $\lfloor K\mu\rfloor$.} and any repetition and recovery order pair $(r,q)$ in the feasible set $\mathcal{R}$ in (\ref{regionR}). Note that each input is computed by at least $r\!-\!(K\!-\!q)$ non-stragglers, so set $\mathcal{R}$ ensures that any subset of $r\!-\!K\!+\!q$ ENs can store at least $m$ rows of $\mathbf{A}$ to multiply each input. The entire scheme including task assignment, input uploading, edge computing and
output downloading are detailed below. 
\subsubsection{Task Assignment}\label{achieveuploadscheme} 
In this paper, we treat tasks from all users equally without considering user priority. So, without
loss of generality, we consider the task assignment of $\{\mathbf{u}_{i,j}\}^{N}_{j=1}$ for user $i\!\in\!\mathcal{M}$. As discussed in Section \ref{keyidea}, for a repetition order $r$, we partition the $N$ input vectors $\{\mathbf{u}_{i,j}\}^{N}_{j=1}$ of each user $i\!\in\!\mathcal{M}$ into $\binom{K}{r}$ equal-sized subsets, each denoted as $\mathcal{U}_{i,\mathcal{K}^{'}}$ and assigned to all the $r$ ENs in subset $\mathcal{K}^{'}\!\subseteq\!\mathcal{K}$ for computation. Each EN $k$ is thus assigned $M\binom{K\!-\!1}{r-\!1}N/\binom{K}{r}\!=\!MNr/K$ inputs corresponding to subsets $\{\mathcal{U}_{i,\mathcal{K}^{'}}\!: i\!\in\!\mathcal{M}, \mathcal{K}^{'}\!\subseteq\!\mathcal{K}, |\mathcal{K}^{'}|\!=\!r, k\!\in\!\mathcal{K}^{'}\}$. By Definition \ref{defenition1}, the \emph{repetition order} is calculated as $K(MNr/K)/MN\!=\!r$, which equals the cardinality $|\mathcal{K}^{'}|$.
\subsubsection{Input Uploading}\label{uploadtime}
Based on the task assignment $\{\mathcal{U}_{i,\mathcal{K}^{'}}\!\}$, each user $i\!\in\!\mathcal{M}$ uploads the subset  $\mathcal{U}_{i,\mathcal{K}^{'}}$ of inputs to the subset $\mathcal{K}^{'}$ of ENs via the uplink channel for $\mathcal{K}^{'}\!\subseteq\!\mathcal{K}$ and $|\mathcal{K}^{'}|\!=\!r$. In other words, each user communicates with all $\binom{K}{r}$ distinct subsets of ENs of cardinality $r$, and any subset of $r$ ENs can form a receiver multicast group. Hence, the resulting uplink channel can be treated as an X-multicast channel with $M$ transmitters, $K$ receivers, and size-$r$ multicast group, the same as that defined in \cite{DOfNiesen} (see Fig. \ref{uplinkdof} for the case with $M\!=\!K\!=\!3$ and $r\!=\!2$). Enabled by asymptotic interference alignment with infinite symbol extensions, each group of $M$ interfering signals from $M$ transmitters can be aligned along the same direction at each receiver\cite{DOfNiesen}. As a result, each receiver can successfully decode a total of $M\binom{K-1}{r-1}$ desired messages from $M$ transmitters over the symbol-extended channel, with the other $M\binom{K-1}{r}$ undesired messages being aligned into $\binom{K-1}{r}$ common subspaces, each subspace containing $M$ undesired messages from $M$ transmitters. For instance, in Fig. \ref{uplinkdof}, each receiver can decode the desired $3\binom{2}{1}\!=\!6$ messages occupying independent subspaces with the undesired 3 signals sent by 3 transmitters being aligned into a common subspace. A per-receiver DoF of $6/7$ can be achieved asymptotically.
\begin{figure}[t]
\centering
\includegraphics[width=2.65in, height=1.5in]{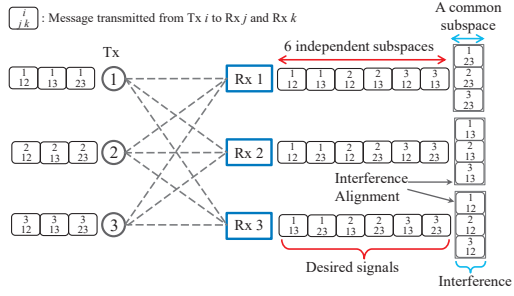}
\vspace{-2mm}
\caption{Interference alignment on the 3-Tx 3-Rx X-multicast channel with size-2 multicast group.}
\vspace{-8mm}
\label{uplinkdof}
\end{figure}
In general, as proved in \cite[Lemma 1]{FanXu} and \cite[Theorem 2]{DOfNiesen}, the optimal per-receiver DoF of this channel is given by $d_{r}^{\textnormal{u}}\!=\!Mr\!/(Mr\!+\!K\!\!-\!r)$.
The per-receiver rate of this channel in the high SNR regime can be approximated as $R^\text{u}_r\!=\!d_{r}^{\textnormal{u}}\!\times\!\log\!P^\textnormal{u}\!+\!o(\log\!P^\textnormal{u})$, where only the first term is relevant in computing (\ref{NULTtau}), so the uploading time can be approximately expressed as $T^\textnormal{u}\!=\!\frac{MNr}{K}nB/(d_{r}^{\textnormal{u}}\log\!P^\textnormal{u}\!+o(\log\!P^\textnormal{u}))$. Let $P^\textnormal{u}\!\to\!\infty$ and $B\!\to\!\infty$, by Definition \ref{defenition2}, the NULT $\tau_{\textnormal{a}}^{\textnormal{u}}$ at repetition order $r$ is given as below,
\begin{equation}\label{tau_u_a2}
\tau_{\textnormal{a}}^{\textnormal{u}}(r)\!=\! \frac{Mr/K}{d_{r}^{\textnormal{u}}}\!=\!\frac{(M\!-\!1)r \!+\! K}{K}.
\end{equation}
\subsubsection{Edge Computing}
After the input uploading phase is finished, each EN computes the products of the assigned input vectors and the stored coded matrix. Following Section \ref{keyidea}, a cascade of an MDS code with rate $1/\rho_1$ and a repetition code with rate $1/\rho_2$ is applied to encode matrix $\mathbf{A}$ into $\mathbf{A}_c$. Under the constraint of the total storage size $K\mu$, the code rates satisfy $\rho_1\rho_2\!=\!K\mu$, $\rho_1\!\in\!\left\{1,K\mu/(K\mu\!-\!1),K\mu/(K\mu\!-\!2),\cdots,K\mu\right\}$ and $\rho_2\!\in\![K\mu]$. Then, we split the coded matrix $\mathbf{A}_c$ into $\binom{K}{\rho_2}$ submatrices $\{\mathbf{A}_{c,\mathcal{K}^{''}}\!\}$, each stored at a distinct subset $\mathcal{K}^{''}\!$ of $\rho_2$ ENs. As shown in Fig. \ref{mds_repetition}, when there are $K\!-q$ stragglers randomly occurring, any subset of $r\!-\!K\!+q$ non-straggling ENs must store at least $m$ encoded rows to compute all outputs. This can be ensured by condition $\rho_1m\!-\!\binom{K\!-(r-\!K\!+q)}{\!\rho_2\!}\rho_1m/\binom{K}{\rho_2}\!\ge\!m$. Further, under this recovery condition, in order to create more data redundancy, the parameter $\rho_2\!\in\![K\mu]$ is maximized as
\begin{align} \label{theoremcoderate2}
\!\!\rho_2\!=\!\inf\!\bigg\{\rho\!:\!\!\binom{K}{\rho}\!-\!\binom{2K\!-\!r\!-\!q}{\rho}\!\ge\! \frac{1}{\rho_1}\binom{K}{\rho},\rho_1\rho\!=\!K\mu,
\rho_1\!\in\!\big\{1,\frac{K\mu}{K\mu\!-\!1},\frac{K\mu}{K\mu\!-\!2}\cdots,K\mu\big\}\!\bigg\};\!\!
\end{align}As an example, in Fig. \ref{downlinkchannel}, for $K\!\!=\!\!M\!\!=\!\!5$, $m\!=\!40$, $\mu\!\!=\!\!3/5$, $q\!\!=\!\!3$, and $r\!\!=\!\!4$, by (\ref{theoremcoderate2}), we have $(\rho_1,\rho_2)\!=\!(3/2,2)$ such that $\mathbf{A}$ is encoded into $60$ rows and then split into $\binom{5}{2}\!\!=\!\!10$ submatrices, each with $6$ rows replicated at $2$ ENs. By the given task input assignment $\{\mathcal{U}_{i,\mathcal{K}^{'}}\}$, each EN $k$ computes $MNr\mu m/K$ row-vector products.  Let $\omega_{1:K}\!\le\!\omega_{2:K}\!\le\!\cdots\!\le\!\omega_{K:K}$ denote the order statistics in a sample of size $K$ from an exponential distribution with mean $1/\eta$\cite{arnold2008first}, by Definition \ref{defenition2}, the NCT is given by
\begin{align}
\tau_{\textnormal{a}}^{\textnormal{c}}(r,q)&\!=\lim\limits_{m\to\infty}\!\frac{\mathbb{E}\left[\frac{MNr\mu m}{K}\omega_{q:K}\right]}{Nm/\eta} 
\!=\!\frac{Mr\mu(H_K\!-\!H_{K-q})}{K},
\end{align}which follows $\mathbb{E}\left[\omega_{q:K}\right]\!=\!(H_K\!-\!H_{K-q})/\eta$\cite[Eq. (4.6.6)]{arnold2008first}.
\begin{figure}[t]
\centering
\includegraphics[width=3.9in, height=2.3in]{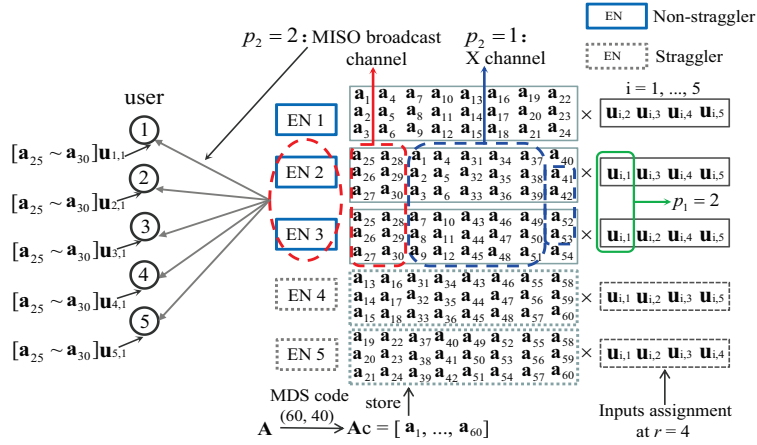}
\vspace{-2mm}
\caption{Illustration of downlink transmission for $K\!=\!M\!=\!5$, $\mu\!=\!3/5$, $m\!=\!40$, $N\!=\!5$, $q\!=\!3$, $r\!=\!4$, and $(\rho_1,\rho_2)\!=\!(3/2,2)$. The MISO broadcast channel and X-channel are formed to transmit outputs of $\{\mathbf{u}_{i,1}\}^5_{i=1}$ back to the users sequentially. This figure only shows the pattern of MISO broadcast channels for transmitting $\{\mathbf{a}_{25}\mathbf{u}_{i,1},\cdots,\mathbf{a}_{30}\mathbf{u}_{i,1}\}^5_{i=1}$. Outputs of $\{\mathbf{u}_{i,2}\}^5_{i=1},\{\mathbf{u}_{i,3}\}^5_{i=1},\cdots,\{\mathbf{u}_{i,5}\}^5_{i=1}$ are transmitted in a similar way.}
\vspace{-8mm}
\label{downlinkchannel}
\end{figure}
\subsubsection{Output Downloading}
At the end of edge computing phase, the $K\!-\!q$ non-straggling ENs return the computed outputs back to users via the downlink channel. Following Section \ref{keyidea}, for each user, the number of input vectors computed by $p_1$ non-straggling ENs equals $\binom{K\!-q}{r-p_1}N/\binom{K}{r}\!=\!B_{p_1}N/\binom{q}{p_1}$, where $r\!-\!(\!K\!-\!q)\!\le\!p_1\!\le\!\min\{r,q\}$. Furthermore, the number of encoded rows of $\mathbf{A}$ replicated at $p_2$ non-straggling ENs is $\binom{K-p_1}{\rho_2-p_2}\rho_1m/\binom{K}{\rho_2}\!=\!B_{p_2}m/\binom{p_1}{p_2}$, where $\max\{\rho_2\!-\! K\!+\!p_1,1\}\!\le\!p_2\!\le\!\min\{p_1,\rho_2\}$. Hence, among the $p_1$ non-straggling ENs, any subset of $p_2$ ENs computing the same $MB_{p_1}NB_{p_2}m/\big(\binom{q}{p_1}\binom{p_1}{p_2}\big) $ outputs can form a transmitter cooperation group, resulting in $\binom{p_1}{p_2}$ groups in total, and each EN cooperation group has outputs to send to all users. The resulting downlink is a cooperative X channel with $p_1$ transmitters, $M$ receivers, and size-$p_2$ cooperation group, as defined in \cite{FanXu} (see Fig. \ref{downlinkdof} for the case with $p_1\!=\!M\!=\!3$ and $p_2\!=\!2$).
As discussed in Sec. \ref{keyidea}, when $p_2\!\ge\!M$,
each subset of $p_2$ ENs can cooperatively transmit common outputs to $M$ users via ZF precoding. In contrast, when $p_2\!<\!M$, each subset of $p_2$ ENs partitions each common output into $\binom{M-1}{p_2-1}$ submessages, and first use ZF precoding to null the interference caused by each submessage at a distinct subset of $p_2\!-\!1$ undesired users. Then, by cascading ZF precoding with asymptotic IA, the rest of interferences from each subset of $t-\!1$ ENs can be aligned into a distinct subspace at each user\cite{FanXu}. Particularly, when $p_2\!=\!M\!-\!1$, each submessage only causes interference to one user, so all interfering signals at each user can be aligned into a common subspace. For example, in Fig. \ref{downlinkdof}, each common message is split into $2$ submessages with each being cancelled at a undesired receiver and causing interference only to another undesired receiver. Then, each receiver can decode the $2\binom{3}{2}\!=\!6$ desired submessages with the rest $2\binom{3}{2}\!=\!6$ interferences being aligned into a common subspace, which achieves a per-receiver DoF of $6/7$.
\begin{figure}[t]
\centering
\includegraphics[width=3.9in, height=2in]{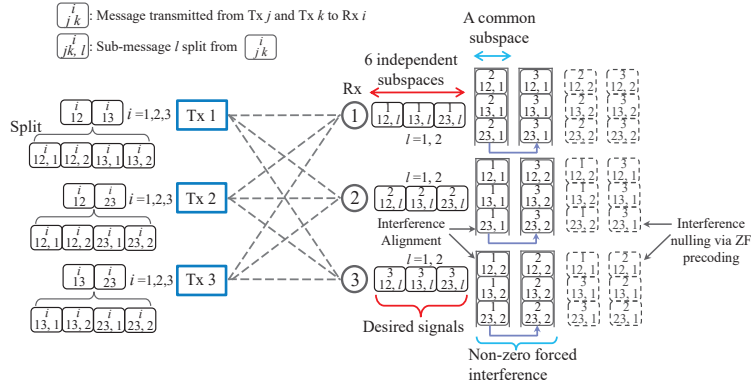}
\vspace{-2mm}
\caption{Interference alignment and ZF precoding on the 3-Tx 3-Rx cooperative X channel with size-2 cooperation group.}
\vspace{-8mm}
\label{downlinkdof}
\end{figure}

In general, by \cite[Lemma 1]{FanXu}, an achievable per-receiver DoF $d_{p_1,M,p_2}^{\textnormal{d}}$ of this downlink channel is given as (\ref{dofdd}), which is within a multiplicative gap of $2$ to the optimal DoF. The per-receiver channel rate for high SNR regime can be approximated as $d_{p_1,M,p_2}^{\textnormal{d}}\!\times\!\log\!P^\textnormal{d}\!+\!o(\log\!P^\textnormal{d})$, where only the first term is relevant in computing (\ref{NDLTtau}). The traffic load for each user to download its desired outputs is $B_{p_1}NB_{p_2}mB/\binom{q}{p_1}$ bits, so the downloading time can be approximately given by   $T^\textnormal{d}\!=\!\frac{B_{p_1}NB_{p_2}mB/\binom{q}{p_1}}{d_{p_1,M,p_2}^{\textnormal{d}}\log\!P^\textnormal{d}+o(\log\!P^\textnormal{d})}$. Let $P^\textnormal{d}\!\to\!\infty$ and $mB\!\to\!\infty$, by Definition \ref{defenition2}, the NDLT for each user to download the outputs replicated at $p_2$ non-stragglers is given by
\begin{equation}
\tau^\textnormal{d}_{p_1,p_2}\!=\!\frac{B_{p_1}B_{p_2}/\binom{q}{p_1}}{d_{p_1,M,p_2}^{\textnormal{d}}}. \label{p_1,p_2}
\end{equation}

Due to the MDS coding, the total number of coded outputs available on the $p_1$ ENs may exceed the number $m$ needed to recover the outputs of each input vector. Denote by $l_{p_1}\!-\!1$ the minimum degrees of replication of needed coded outputs on the $p_1$ ENs, $l_{p_1}$ is determined by \begin{small}$l_{p_1} \!=\! \inf\!\big\{l\!:\!\sum\nolimits^{\min\{p_1,\rho_2\}}_{p_2=l}\!\!B_{p_2}m\!\le\!m, l\!\ge\!\max\{\rho_2\!-\! K\!+\!p_1,1\} \big\}$\end{small}, so the number of needed coded outputs replicated at $l_{p_1}\!\!-\!1$ ENs equals $MB_{p_1}NB_{l_{p_1}\!-1}m/\binom{q}{p_1}$, where $B_{l_{p_1}-1}\!=\!1\!-\!\!\sum^{l_\textnormal{max}}_{p_2=l_{p_1}}\!\!B_{p_2}$. Note that   $B_{l_{p_1}\!-\!1}m/\binom{p_1}{l_{p_1}\!-1}$ can be seen as an integer for infinitely large $m$ since $\big(\!B_{l_{p_1}\!-1}m\!\!\!\mod{\!\binom{p_1}{l_{p_1}\!-1}}\big)\!/m\!<\!\binom{p_1}{l_{p_1}\!-1}/m\!\to\! 0$ as $m\!\to\!\infty$. So it enables any subset of $l_{p_1}\!-\!1$ ENs among $p_1$ ENs to cooperatively transmit $B_{p_1}NB_{l_{p_1}\!-1}m/\big(\binom{q}{p_1}\binom{p_1}{l_{p_1}\!-1}\big)$ common outputs to each user, the downlink channel is also a cooperative X channel with $p_1$ transmitters, $M$ receivers, and size-$(l_{p_1}\!-\!1)$ cooperation group. Similar to (\ref{p_1,p_2}), the NDLT for each user to download the outputs replicated at $l_{p_1}\!\!-\!1$ non-stragglers is given by
\begin{equation}
\tau_{p_1,l_{p_1}-1}^{\textnormal{d}}\!=\!\frac{B_{p_1}B_{l_{p_1}-1}/\binom{q}{p_1}}{d_{p_1,M,l_{p_1}-1}^{\textnormal{d}}}.\label{download2}
\end{equation}
Furthermore, by considering all the inputs computed by $p_1$ ENs,  with $p_1$ from $r\!-\!(K\!-\!q)$ to $\min\{r,q\}$, and all the outputs replicated at $p_2$ ENs, with $p_2$ from $l_{p_1-1}$ to $\min\{p_1,\rho_2\}$, and by summing all download time given in (\ref{p_1,p_2}) and (\ref{download2}), the NDLT $\tau_{\textnormal{a}}^{\textnormal{d}}(r,q)$ is obtained as below,
\begin{equation}
\tau_{\textnormal{a}}^{\textnormal{d}}(r,q)\!=\!\sum\limits^{\min\{r,q\}}_{p_1\!=r-\!K\!+q}\!\!\!\!\!B_{p_1}\!\!\left(\sum\limits^{l_\textnormal{max}}_{p_2=l_{p_1}}\!\!\frac{B_{p_2}}{d_{p_1,M,p_2}^{\textnormal{d}}}\!+\!
\!\frac{B_{l_{p_1}\!-1}}{d_{p_1,M,l_{p_1}\!-1}^{\textnormal{d}}}\!\right)
\end{equation}

We now illustrate the output downloading latency by the example in Fig. \ref{downlinkchannel}. First, for inputs $\{\mathbf{u}_{i,1}\}^5_{i=1}$ computed by $p_1\!\!=\!\!2$ ENs, there are 
$30$ outputs $\{\mathbf{a}_{25}\mathbf{u}_{i,1},\ldots,\mathbf{a}_{30}\mathbf{u}_{i,1}\}^5_{i=1}$ replicated at  $p_2\!=\!2$ ENs. These $30$ outputs can be cooperatively transmitted back to the users via ZF precoding, resulting in a 2-transmitter 5-receiver MISO broadcast channel that is a special case of cooperative X channels under full transmitter cooperation. As a result, an NDLT of $3/40$ is achieved. After this round of transmission, users still need $34\!\times\!5\!=\!170$ outputs inside the blue dashed rectangle in Fig. \ref{downlinkchannel}, which can be transmitted by the 2 ENs via interference alignment. The downlink is a 2-transmitter 5-receiver X channel that is a special case of cooperative X channels with size-$1$ cooperation group, yielding the NDLT of  $51/100$. Thus, the NDLT for outputs of $\{\mathbf{u}_{i,1}\}^5_{i=1}$ is $3/40\!+\!51/100\!=\!117/200$. Then, the input vectors $\{\mathbf{u}_{i,2}\}^5_{i=1},\{\mathbf{u}_{i,3}\}^5_{i=1}$ are also computed by $p_1\!\!=\!\!2$ ENs, their outputs can be transmitted in a similar way, which achieves an NDLT of $(117/200)\!\times\!2\!=\!117/100$. Likewise, for the inputs $\{\mathbf{u}_{i,4}\}^5_{i=1},\{\mathbf{u}_{i,5}\}^5_{i=1}$ computed by $p_1\!\!=\!\!3$ ENs, the 3-transmitter 5-receiver cooperative X-channel with size-$2$ cooperation group, and 3-transmitter 5-receiver X-channel are formed to transmit the total $400$ outputs, yielding an NDLT of $(21/100\!+\!77/300)\!\times\! 2\!=\!14/15$. Thus, in this example, the total NDLT at $(r,q)\!=\!(4,3)$ is $14/15\!+\!(117/200)\!\times\!3\!=\!1613/600$.
\subsubsection{Inner Bound of Compute-Download Latency Region}
For an NULT $\tau^\textnormal{u}\!=\!\tau_{\textnormal{a}}^{\textnormal{u}}(r)$ given in (\ref{tau_u_a2}) for some $r\!\in\![K]$, by the region $\mathcal{R}$ given in (\ref{regionR}), the feasible recovery order $q$ satisfies $\lceil 1/\mu\rceil\!+\!K\!-r\!\le\!q\!\le\!K$. By Remark \ref{remarkconvex}, for any two integer-valued $q_1$ and $q_2$ in $\big[\lceil 1/\mu\rceil\!+\!K\!-\!r\!:\! K\big]$, any convex combination of achievable pairs $(\tau_{\textnormal{a}}^{\textnormal{c}}(r,q_1),\tau_{\textnormal{a}}^{\textnormal{d}}(r,q_1)$ and $(\tau_{\textnormal{a}}^{\textnormal{c}}(r,q_2),\tau_{\textnormal{a}}^{\textnormal{d}}(r,q_2))$ can also be achieved. So an inner bound $\mathscr{T}_{in}(\tau^\textnormal{u})$ of the compute-download latency region is given as the convex hull of set $\big\{\!\left(\tau_{\textnormal{a}}^{\textnormal{c}}(r,q),\tau_{\textnormal{a}}^{\textnormal{d}}(r,q)\right)\!:\!q\!\in\!\big[\lceil\frac{1}{\mu}\rceil\!+\!K\!-\!r\!:\! K\big]\big\}$.

\section{Numerical Examples and Discussion} \label{simulation}
In this section, we first evaluate the system performance in terms of the asymptotic end-to-end execution time $\tau$ in Eq. (\ref{totaltime}). Then, we use numerical examples to show the average uploading, computing, downloading, and end-to-end execution times in the non-asymptotic regime.
\vspace{-3mm}
\subsection{Asymptotic Results}
The analysis in the previous section has shown that, by choosing the repetition order $r$ and the recovery order $q$, one can obtain different triplets of the upload latency $\tau^\textnormal{u}$, computation latency $\tau^\textnormal{d}$, and download latency $\tau^\textnormal{c}$. As a result, parameters $(r,q)$ can be optimized to minimize the end-to-end execution time, yielding the minimum end-to-end execution time
$\tau^*\!=\! \min\limits_{(r,q)\in\mathcal{R}} \tau(r,q) $.
\begin{figure*}
\centering
\subfigure[$\tau^*$ versus $\delta_c$. $\delta_d\!=\!8$.]{
\label{computingSpeed}
  \includegraphics[width=2.8in, height=2.1in]{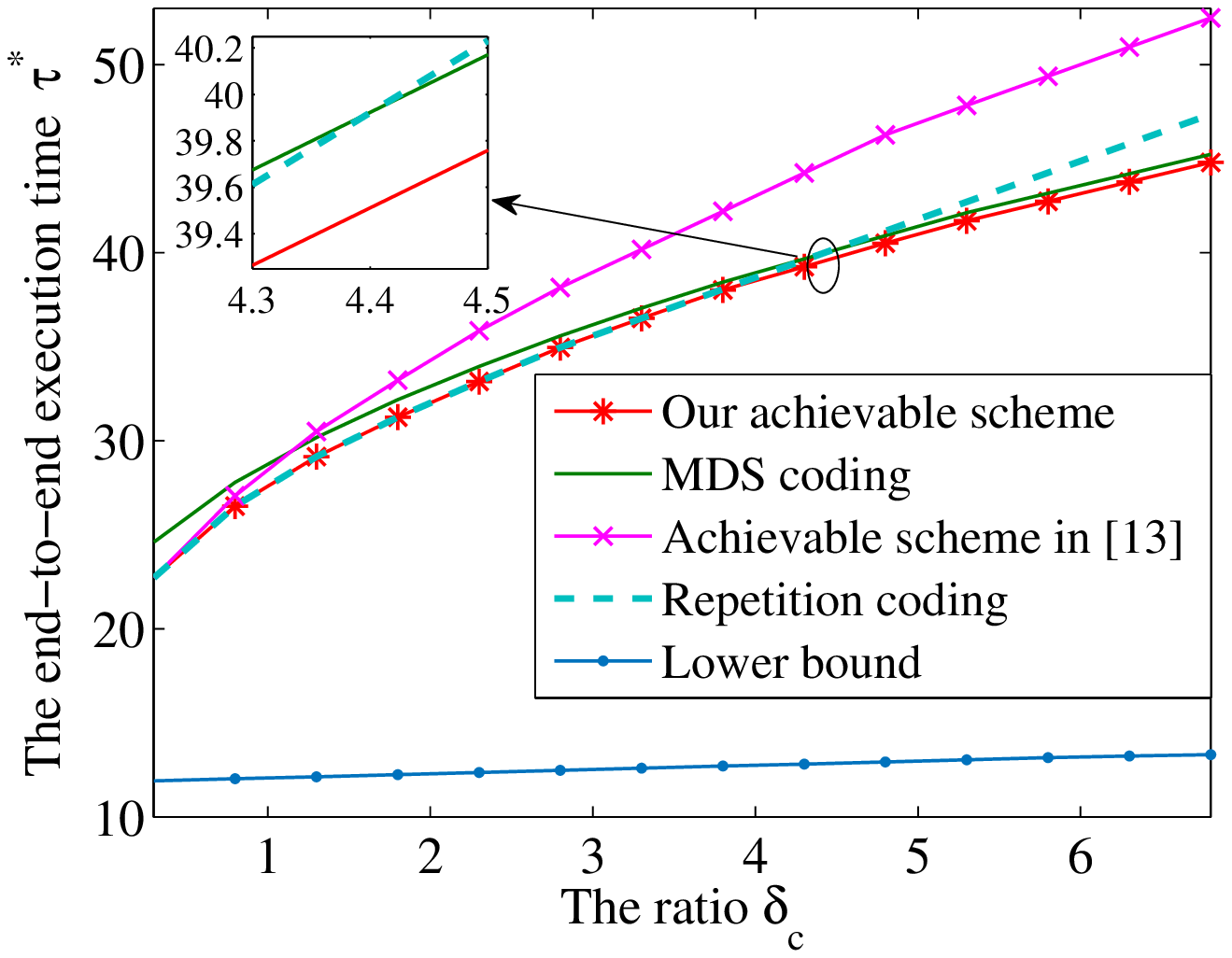}}
  \vspace{-2mm}\hspace{1mm}
 \subfigure[$\tau^*$ versus $\delta_d$. $\delta_c\!=\!5$.]{
\label{outputsize}
  \includegraphics[width=2.8in, height=2.1in]{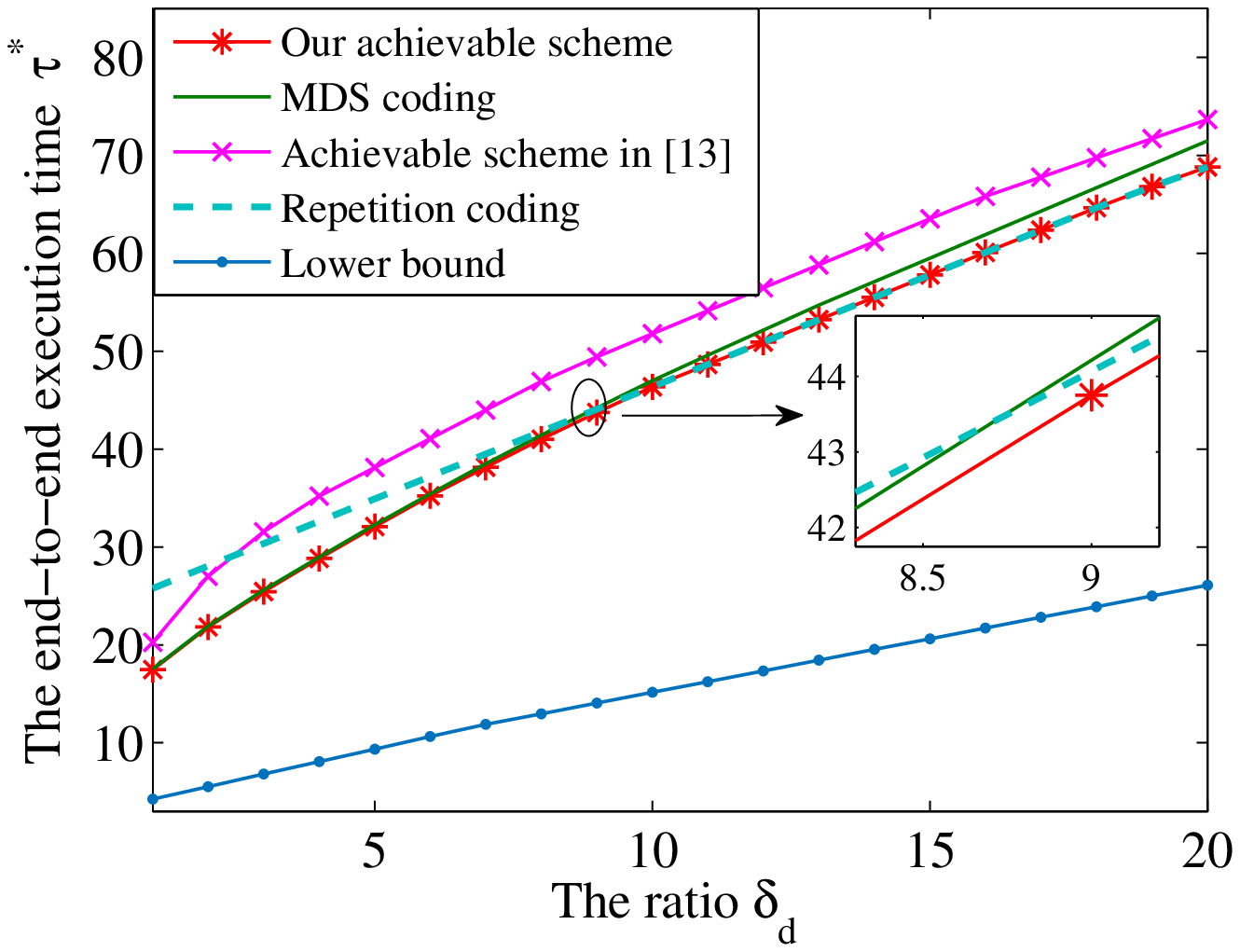}}
   \vspace{-1mm}
\caption{The impacts of the ratios $\delta_c$ and $\delta_d$ on the end-to-end execution time $\tau^*$.} \label{netparameters}
\vspace{-8mm}
\end{figure*}
To illustrate the minimum end-to-end execution time $\tau^*$, we consider a MEC network with $M\!=\!8$ users and $K\!=\!10$ servers. Each EN has a fractional storage size of $\mu\!=\!3/5$. We compare the proposed scheme with the following baseline strategies, for which parameters $r$ and $q$ are also optimized:
a) MDS coding: Only an MDS code is applied to encode $\mathbf{A}$, i.e., we have the special case $(1/\rho_1,1/\rho_2)\!=\!(1/K\mu,1)$;
b) Repetition coding: Only the repetition code is applied to encode $\mathbf{A}$, i.e., we have  $(1/\rho_1,1/\rho_2)$ $=(1,1/K\mu)$;
c) Achievable schemes in \cite{jingjing}: Each EN has the inputs of all users and transmits the outputs back by using one-shot linear precoding. A cascade of MDS and repetition codes is applied to encode $\mathbf{A}$ as in the proposed scheme.

Fig. \ref{netparameters} shows the impacts of the ratio $\delta_c$ between computation and upload reference latencies and the ratio $\delta_d$ between download and upload reference latencies on the end-to-end execution time $\tau^*$. By providing more flexible choices for task uploading, the proposed scheme increases the achievable region in terms of the latencies for task uploading, computing, and output downloading, leading to a reduction in the end-to-end execution time. In Fig. \ref{netparameters}-(a), it is also observed that, when the computing speed is slow, i.e., $\delta_c$ is large, it is suboptimal to upload the input data of each user to all ENs, and, as a result, the proposed scheme obtains gains with respect to \cite{jingjing} that increases with $\delta_c$. In contrast, when $\delta_c$ is small, the proposed scheme reduces to that of \cite{jingjing}, since, for fast server computing speeds, the optimal task assignment policy prefers to replicate tasks at all ENs.

Fig. \ref{netparameters} also demonstrates the advantages of using both MDS and repetition coding. In Fig. \ref{netparameters}-(b), when $\delta_d$ is sufficiently small, i.e., $\delta_d\!\le\! 8.5$, the proposed cascaded MDS-repetition coding scheme approaches MDS coding. This is because in this regime, repetition coding only brings limited transmission cooperation gain in the downlink, but it requires large upload and computation latencies for output recovery. Therefore, when uplink transmission and edge computing are the major bottlenecks of the offloading process, MDS coding can reduce the input uploading constraint required for output recovery and  mitigate random straggling effects at computing phases. In contrast, when $\delta_d$ is larger, downlink latency becomes the bottleneck, and repetition coding is preferable since it can fully exploit transmission cooperation gains in data downloading to reduce the download time.
\vspace{-4mm}
\subsection{Non-asymptotic Results}
We consider a 2-user 3-server MEC network, and set the network parameters $N\!=\!3$, $\mu\!=\!2/3$, $n\!=\!100$, $m\!=\!900$ rows, $B\!=\!8$ bits, $1/\eta\!=\!10^{-4}$ seconds. The uplink and downlink channel bandwidths are $B_u\!=\!B_d\!=\!$ 100 KHz. We consider the normalized Rayleigh channel fading and normalized noise power. By set $\mathcal{R}$ in (\ref{regionR}), we have feasible pairs $(r,q)\!\in\!\left\{(2,3), (3,3), (3,2)\right\}$. Interference alignment on $2\!\times\!3$ or $3\!\times\!2$ networks only needs finite symbol extensions over these channels.  The simulations are averaged over 50000 independent channel realizations.

Table \ref{tablee} shows the actual average input uploading time $T^\text{u}$, edge computing time $T^\text{c}$, output downloading time $T^\text{d}$, and end-to-end execution time $T$, at different transmission powers. It is seen that the optimal policy is to upload user inputs to all 3 ENs and then wait for the fastest 2 ENs to finish their tasks, which significantly reduces the total time $T$. For example, when $P^\text{u}\!=\!P^\text{d}\!=\!20$ dB, the total time at $(r,q)\!=\!(2,3)$ is decreased by $23\%$ compared to the optimal time at $(r,q)\!=\!(3,2)$. In this policy, the higher repetition order enables the transmission cooperation to reduce downloading times, and the lower recovery order mitigates the effect of 1 straggling node and thus reduces computing times. Comparing the times at $(r,q)\!=\!(2,3)$ and $(r,q)\!=\!(3,3)$, we also see that replicating inputs at all 3 ENs may cause more total times since the gain on reducing downloading times is limited compared to the increased computing time.
\begin{small}\begin{table}\small
{\caption{Actual average uploading, computing, and downloading times}  \label{tablee}
\begin{center}
\vspace{-7mm}
\begin{tabular}{c|ccc|ccc|ccc}
\hline
\multicolumn{1}{c|}{$P^\text{u}$ (or $P^\text{d}$) }  &\multicolumn{3}{c|}{$10$ dB} &\multicolumn{3}{c|}{$20$ dB} & \multicolumn{3}{c}{$30$ dB} \\
\cline{1-10}
\multicolumn{1}{c|}{ $(r,q)$}
&  $(2,3)$     & $(3,3)$       & $(3,2)$  &  $(2,3)$     & $(3,3)$       & $(3,2)$ &  $(2,3)$     & $(3,3)$       & $(3,2)$     \\
\hline
$T^\text{u}$ (sec.)           &0.112                           &0.158                     & 0.158          &0.027                           &0.036                     & 0.036         &0.011                           &0.014                     & 0.014     \\
$T^\text{c}$ (sec.)            &0.44                          &0.66                    & 0.3                   & 0.44           & 0.66           & 0.3          & 0.44           & 0.66           & 0.3          \\
$T^\text{d}$ (sec.)           &0.315                          & 0.106                    & 0.313                   &0.105           & 0.048            & 0.105          &0.053           & 0.029           &0.053          \\
Total time $T$ (sec.)             &0.867                          &0.924                     &0.771                  & 0.572          &0.744         &0.441        &0.504           & 0.703         &0.367            \\\hline
\end{tabular}
\end{center}}
\vspace{-12mm}
\end{table}\end{small}

\section{Conclusions}\label{conclusion}
This paper studies the communication-computation tradeoff for distributed matrix multiplication in multi-user multi-server MEC networks with straggling ENs. We propose a new task offloading policy that leverages cascaded coded computing and cooperative transmission to alleviate the impact of straggling ENs and speed up the communication phase. We derive achievable upload-compute-download latency triplets and the end-to-end execution time, as well as the lower bounds. We prove that the obtained upload latency is optimal for fixed computation and download latencies, and that the computation latency and download latency are within constant multiplicative gaps to their respective lower bounds for a sufficiently large upload latency. Our results reveal that for a fixed upload latency, the download latency can be traded for computation latency; and that increasing the upload latency can reduce both the computation latency and download latency.
Through numerical results, we show that the proposed policy is able to obtain a more flexible trade-off among upload, computation, and download latencies than baseline schemes, and that this leads to a significant reduction in the end-to-end execution time.

\section*{Appendix: Proofs of Converse} \label{converseproof} 

In this appendix, we prove Theorem \ref{lower_bound} and Lemma \ref{gapLemma}. Note that in each subsection, we first derive the lower bound in Theorem \ref{lower_bound}, and then prove the multiplicative gap in Lemma \ref{gapLemma}. 

For a repetition-recovery order pair $(r,q)$, as discussed, each input is computed by at least  $r\!-\!(K\!-\!q)$ non-stragglers. The condition $(r-\!K\!+q)\mu\!\ge\!1$ must be satisfied such that any subset of $r\!-\!K\!+\!q$ non-stragglers are able to provide sufficient information to compute the outputs of all users. This proves that no pair $(r,q)$ is feasible outside the feasible set $\mathcal{R}$ in (\ref{regionR}).
Then, we consider an arbitrary user input assignment policy $\big\{\mathcal{U}_{i,\mathcal{K}^{'}}\!: i\!\in\!\mathcal{M},\mathcal{K}^{'}\!\subseteq\!\mathcal{K}, |\mathcal{K}^{'}|\!=\!r\big\}$ with $(r,q)\!\in\!\mathcal{R}$. The input vectors from user $i$ assigned to EN $k$ are denoted as set $\mathcal{I}_{i,k}\!\triangleq\!\big\{\mathcal{U}_{i, \mathcal{K}^{'}}\!\big\}_{\mathcal{K}^{'}\subset\mathcal{K}:k\in\mathcal{K}^{'}}$ for $i\!\in\!\mathcal{M}$ and $k\!\in\!\mathcal{K}$. The size of $\mathcal{I}_{i,k}$ is denoted as $\gamma_{i,k}NnB$ bits, where the ratio $\gamma_{i,k}$ satisfies
\begin{align}
&\sum\limits_{k\in\mathcal{K}}\gamma_{i,k} = r ,~ i \in\mathcal{M} \label{cons111}\\
&~0\le \gamma_{i,k} \le 1,~i \in\mathcal{M},~\text{and}~k\in\mathcal{K}. \label{cons222}
\end{align}In the following, we first derive the lower bounds on the NULT, NCT, and NDLT for a particular task assignment policy $\big\{\mathcal{U}_{i,\mathcal{K}^{'}}\!\big\}$ with repetition-recovery order pair $(r,q)\!\in\!\mathcal{R}$. Then, by considering all possible task assignment policies and the effect of random stragglers, we obtain the minimum lower bounds for the NULT $\tau_l^{\textnormal{u}}$, NCT $\tau_l^{\textnormal{c}}$, and NDLT $\tau_l^{\textnormal{d}}$. For a fixed NULT at $r\!\in\!\big[\lceil\frac{1}{\mu}\rceil,K\big]$, by convexity of the compute-download latency region, an outer bound of this region is given by the convex hull of all pairs $\{(\tau_l^{\textnormal{c}},\tau_l^{\textnormal{d}})\}$, as described in Theorem \ref{lower_bound}.
\subsection{Lower Bound and Optimality of NULT}
\subsubsection{Lower bound}For a particular task assignment policy $\big\{\mathcal{U}_{i,\mathcal{K}^{'}}\!\big\}$, we use genie-aided arguments to derive a lower bound on the NULT. Specifically, for any EN $k$ and user $i_o$, consider the following three disjoint subsets of task input vectors (or messages):
\begin{align}
\mathcal{W}_{r} &= \{\mathcal{U}_{i, \mathcal{K}^{'} }: i\in\mathcal{M}, k\in\mathcal{K}^{'}\},\\
\mathcal{W}_{t} &= \{\mathcal{U}_{i, \mathcal{K}^{'} }:i=i_o,k\notin \mathcal{K}^{'} \},\\
\overline{\mathcal{W}} &= \{\mathcal{U}_{i,\mathcal{K}^{'}}: i\ne i_o~\text{and}~k\notin \mathcal{K}^{'}\}.
\end{align}
The set $\mathcal{W}_{r}$ indicates the input messages from all users assigned to EN $k$ or all input messages that EN $k$ needs to decode, which satisfies $|\mathcal{W}_{r}|\!=\!\sum_{i\in\mathcal{M}}\gamma_{i,k}NnB$. The set $\mathcal{W}_{t}$ indicates the input messages from user $i_o$ assigned to all ENs in $\mathcal{K}$ excluding EN $k$, which satisfies $|\mathcal{W}_{t}|\!=\!(1\!-\!\gamma_{i_o,k})NnB$. The last set $\overline{\mathcal{W}}$ indicates all input messages from users in $\mathcal{M}$ excluding user $i$ assigned to ENs in $\mathcal{K}$ excluding EN $k$.

Let a genie provide the messages $\overline{\mathcal{W}}$ to all ENs, and additionally provide messages $\mathcal{W}_{r}$ to ENs in $\mathcal{M}/\{k\}$.
The received signal of EN $j$ can be represented as
\begin{align}
\mathbf{y}_j &= \sum\limits^{M}_{i=1,i\ne i_o}\mathbf{H}_{ji}^{\textnormal{u}}\mathbf{X}_{i} + \mathbf{H}_{ji_o}^{\textnormal{u}}\mathbf{X}_{i_o}+\mathbf{Z}_j^{\textnormal{u}},
\end{align}
where the diagonal matrices $\mathbf{H}_{ji}^{\textnormal{u}}$, $\mathbf{X}_{i}$, and $\mathbf{Z}_j^{\textnormal{u}}$ denotes the channel coefficients from user $i$ to EN $j$, signal transmitted by user $i$, and noise received at EN $j$, respectively, over the block length $T^\textnormal{u}$.
The ENs in $\mathcal{M}/\{k\}$ have messages $\overline{\mathcal{W}}\!\bigcup\!\mathcal{W}_{r}$, which include the input messages that EN $k$ should decode and input messages transmitted by all users excluding user $i_o$. By this genie-aided information, each EN $j\!\in\!\mathcal{M}/\{k\}$ can construct the transmitted symbols $\{\mathbf{X}_i\!:i\!\ne\!i_o\}$ and subtract them from the received signal. So we can rewrite the signal received at EN $j\!\ne\!k$ as
\begin{equation} \label{equation1122}
\bar{\mathbf{y}}_j = \mathbf{y}_j - \sum\limits_{i\in\mathcal{M}/\{i_o\}}\mathbf{H}_{ji}^{\textnormal{u}}\mathbf{X}_{i} = \mathbf{H}_{ji_o}^{\textnormal{u}}\mathbf{X}_{i_o} + \mathbf{Z}_j^{\textnormal{u}}.
\end{equation}
Each EN $j\!\in\!\mathcal{M}/\{k\}$ needs to decode the input messages in subset $\mathcal{W}_{t}$ assigned to it, denoted as $\mathcal{W}^j_{t}$. By Fano{'}s inequality and (\ref{equation1122}), we have
\begin{equation} \label{fano1}
H(\mathcal{W}^j_{t}|\mathbf{y}_j,\overline{\mathcal{W}},\mathcal{W}_{r}) \le T^\textnormal{u} \epsilon, ~~j\in\mathcal{M}/\{i\}.
\end{equation}
Since EN $k$ can decode input messages $\mathcal{W}_{r}$ assigned to it, by Fano{'}s inequality, we also obtain
\begin{equation}\label{Fano11}
H(\mathcal{W}_{r}|\hat{\mathbf{y}}_k,\overline{\mathcal{W}}) \le  T^\textnormal{u}\epsilon .
\end{equation}
Then, EN $k$ can construct the transmitted symbols $\{\mathbf{X}_i\!:i\!\ne\!i_o\}$ based on genie-aided messages $\overline{\mathcal{W}}$ and its decoded messages $\mathcal{W}_{r}$, and subtract them from its received signal, obtaining
\begin{equation}
\bar{\mathbf{y}}_k = \mathbf{y}_k - \sum\limits_{i\in\mathcal{M}/\{i_o\}}\mathbf{H}_{ki}^{\textnormal{u}}\mathbf{X}_{i} = \mathbf{H}_{ki_o}^{\textnormal{u}}\mathbf{X}_{i_o} + \mathbf{Z}_k^{\textnormal{u}}.
\end{equation}
Reducing the noise in the constructed signal $\bar{\mathbf{y}}_k$ and multiplying it by $\mathbf{H}_{ji_o}^{\textnormal{u}}\left(\mathbf{H}_{ki_o}^{\textnormal{u}}\right)^{-1}$, we obtain
\begin{equation}
\bar{\mathbf{y}}^{j}_k  = \mathbf{H}_{ji_o}^{\textnormal{u}}\left(\mathbf{H}_{ki_o}^{\textnormal{u}}\right)^{-1} \bar{\mathbf{y}}_k =  \mathbf{H}_{ji_o}^{\textnormal{u}}\mathbf{X}_{i_o} + \hat{\mathbf{Z}}_j^{\textnormal{u}},
\end{equation}
where $\hat{\mathbf{Z}}_j^{\textnormal{u}}$ is the reduced noise. By (\ref{equation1122}), we see that $\bar{\mathbf{y}}^{j}_k $ is a degraded version of $\bar{\mathbf{y}}_j $ for EN $j\!\in\!\mathcal{M}/\{i\}$. Hence, for the messages that ENs in $\mathcal{M}/\{i\}$ can decode, EN $k$ must also be able to decode them, and we have
\begin{equation}\label{fano2}
\!\!\!H(\mathcal{W}^j_{t}|\hat{\mathbf{y}}_k,\!\overline{\mathcal{W}},\!\mathcal{W}_{r})\!\le\! H(\mathcal{W}^j_{t}|\mathbf{y}_j,\!\overline{\mathcal{W}},\!\mathcal{W}_{r})\! \le\! T^\textnormal{u} \epsilon, j\!\!\in\!\!\mathcal{M}\!/\!\{i\}\!.\!\!
\end{equation}
Using genie-aided information, receiver cooperation, and noise reducing as discussed above can only improve channel capacity. Thus, we obtain the following chain of inequalities,
\begin{align}
|\mathcal{W}_{r}|+|\mathcal{W}_{t}| &= H(\mathcal{W}_{r},\mathcal{W}_{t}) \nonumber\\
&\stackrel{(a)}{=}\!H(\mathcal{W}_{r},\mathcal{W}_{t}|\overline{\mathcal{W}})\nonumber\\
&\stackrel{(b)}{=}\!I(\mathcal{W}_{r},\mathcal{W}_{t};\hat{\mathbf{y}}_k|\overline{\mathcal{W}}) + H(\mathcal{W}_{r},\mathcal{W}_{t}|\hat{\mathbf{y}}_k,\overline{\mathcal{W}})\nonumber\\
&\stackrel{(c)}{=}\!I(\mathcal{W}_{r},\mathcal{W}_{t};\hat{\mathbf{y}}_k|\overline{\mathcal{W}}) + H({\mathcal{W}_{r}|\hat{\mathbf{y}}_k,\overline{\mathcal{W}}})+H({\mathcal{W}_{t}|\hat{\mathbf{y}}_k,\mathcal{W}_{r},\overline{\mathcal{W}}})\nonumber\\
& \le I(\mathcal{W}_{r},\mathcal{W}_{t};\hat{\mathbf{y}}_k|\overline{\mathcal{W}})+ H({\mathcal{W}_{r}|\hat{\mathbf{y}}_k,\overline{\mathcal{W}}})\!+\! \sum\nolimits_{j\in\mathcal{M}/\{k\}}\!\!\!\!\!H({\mathcal{W}^j_{t}|\hat{\mathbf{y}}_k,\mathcal{W}_{r},\overline{\mathcal{W}}})\nonumber\\
&\stackrel{(d)}{\le} I(\mathcal{W}_{r},\mathcal{W}_{t};\hat{\mathbf{y}}_k|\overline{\mathcal{W}}) +  T^\textnormal{u} \epsilon+ \sum\nolimits_{j\in\mathcal{M}/\{k\}} T^\textnormal{u}\epsilon\nonumber\\
&\stackrel{(e)}{\le}  I(\mathbf{x}_1,\mathbf{x}_2,\cdots,\mathbf{x}_{a_i},\mathbf{x}_{i_o};\hat{\mathbf{y}}_i|\overline{\mathcal{W}}) + M T^\textnormal{u}\epsilon \nonumber\\
&\stackrel{(f)}{\le}T^\textnormal{u} \log P^\textnormal{u}  + M T^\textnormal{u}\epsilon,  \label{inequality}
\end{align}
where (a) is due to the independence of messages; (b) and (c) are based on the chain rule; (d) follows Fano{'}s inequalities (\ref{Fano11}) and (\ref{fano2}); (e) uses the data processing inequality; and (f) follows the DoF bound of MAC channel. Dividing (\ref{inequality}) by $NnB\!/\!\log\!P^\textnormal{u}$, and let $P^\textnormal{u}\!\to\!\infty$ and $\epsilon\!\to\!0$ as $B\to\infty$, we have
\begin{align}
\!\!\!\tau^\textnormal{u}\!&\ge\!\frac{|\mathcal{W}_{r}|\!+\!|\mathcal{W}_{t}|}{NnB}\!=\!\sum\limits_{i\in\mathcal{M}}\!\!\gamma_{i,k} \!+\! 1\!-\gamma_{i_o,k}\! = \!\sum_{i\in\mathcal{M}/\{i_o\}}\!\!\!\gamma_{i,k}\!  +\!1.\!\!
\end{align}
Hence, the NULT for a particular task assignment $\boldsymbol{\gamma}\!\triangleq\![\gamma_{i,k}]_{i\in\mathcal{M},k\in\mathcal{K}}$ satisfies $\tau^\textnormal{u} \!\ge\! \sum_{i\in\mathcal{M}/\{i_o\}}\!\gamma_{i,k} \!+\! 1$ for $k\!\in\!\mathcal{K},i_o\!\in\!\mathcal{M}$, i.e., the minimum NULT for task assignment policy $\boldsymbol{\gamma}$ is lower bounded by
\begin{equation}
\tau^{\textnormal{u}^*}(r,\boldsymbol{\gamma}) \ge \max\limits_{ k\in\mathcal{K}, i_o\in\mathcal{M}}~ \sum_{i\in\mathcal{M}/\{i_o\}}\gamma_{i,k}  + 1.
\end{equation}
Further, the minimum NULT over all feasible task assignment is given as $\tau^{\textnormal{u}^*}\!(r)\!=\!\min\limits_{\boldsymbol{\gamma}}\tau^{\textnormal{u}^*}\!(r,\boldsymbol{\gamma})$, i.e., it can be lower bounded by the optimal solution of the optimization problem
\begin{align}
\mathcal{P}_1:\quad&\min\limits_{\boldsymbol{\gamma}}~\max\limits_{ k\in\mathcal{K}, i_o\in\mathcal{M}}~ \sum_{i\in\mathcal{M}/\{i_o\}}\gamma_{i,k}  + 1 \nonumber\\
\mathnormal{s.t.}&\quad (\ref{cons111}), (\ref{cons222}).\nonumber
\end{align}
Note that (\ref{cons111}) and (\ref{cons222}) are the task assignment constraints for recovery order $r$. By defining a new variable $\lambda_{k,\bar{i}_o}\!=\!\sum\limits_{i\in\mathcal{M}/\{i_o\}}\gamma_{i,k}$, Problem $\mathcal{P}_1$ can be transformed into
\begin{align}
\mathcal{P}_2:\quad&\min\limits_{\boldsymbol{\lambda}}~\max\limits_{ k\in\mathcal{K}, i_o\in\mathcal{M}} \lambda_{k,\bar{i}_o} + 1 \nonumber\\
\mathnormal{s.t.}&\quad \sum_{k \in\mathcal{K}}\lambda_{k,\bar{i}_o} = r(M-1), ~ i_o\in\mathcal{M}, \label{P2con1}\\
                 &\quad ~0\le \lambda_{k,\bar{i}_o} \le M-1,~ k \in\mathcal{K},~i_o\in\mathcal{M}.
\end{align}
\begin{lemma}\label{lemma2}The unique optimal solution to $\mathcal{P}_2$ is given by $\lambda^*_{k,\bar{i}_o}\!=\!r(M\!-\!1)/K$, $k\!\in\!\mathcal{K}$, $i_o\!\in\!\mathcal{M}$.\end{lemma}
\emph{Proof:}
By contradiction, assuming that there exists an optimal solution $\{\lambda^{'}_{k,\bar{i}_o}\}$ to $\mathcal{P}_2$ which does not satisfy $\lambda^{'}_{k,\bar{i}_o}\!=\!r(M\!-\!1)/K$ for $k\!\in\!\mathcal{K}$ and $i_o\!\in\!\mathcal{M}$. By (\ref{P2con1}), there must exist index $j\!\in\!\mathcal{K}$ such that $\lambda^{'}_{j,\bar{i}_o}\!>\!r(M\!-\!1)/K$ for $i_o\!\in\!\mathcal{M}$; otherwise, we have $\sum\limits_{k \in\mathcal{K}}\lambda^{'}_{k,\bar{i}_o}\!<\!r(M-1)$ for $i_o\!\in\!\mathcal{M}$. The optimal objective satisfies
$\max\limits_{ k\in\mathcal{K}, i_o\in\mathcal{M}} \lambda^{'}_{k,\bar{i}_o} \!+\! 1\!\ge\!\lambda^{'}_{j,\bar{i}_o} \!+ \! 1>r(M\!-\!1)/K\!+\!1$, and
$r(M\!-\!1)/K\!+\!1$ is the objective value at $\lambda_{k,\bar{i}_o}\!=\!r(M\!-\!1)/K$, $k\!\in\!\mathcal{K}$, $i_o\!\in\!\mathcal{M}$. So the initial assumption does not hold. The optimal solution to $\mathcal{P}_2$ is $\lambda^*_{k,\bar{i}_o}\!=\!r(M\!-\!1)/K$ for $k\!\in\!\mathcal{K}$ and $i_o\!\in\!\mathcal{M}$, which is unique.

In turn, we use $\{\lambda^*_{k,\bar{i}_o}\}$ in Lemma \ref{lemma2} to construct a feasible solution to $\mathcal{P}_1$ by letting $\gamma^*_{i,k}\!=\!\lambda^*_{k,\bar{i}_o}/(M\!-\!1)$ for $i\!\in\!\mathcal{M}$ and $k\!\in\!\mathcal{K}$, and hence obtain the optimal solution to $\mathcal{P}_1$ as $\gamma^*_{i,k}\!=\!r/K$. Therefore, at repetition order $r$, the minimum NULT $\tau^{\textnormal{u}^*}(r)$ is lower bounded by
\begin{equation} \label{lowerbound_tauu}
\tau^{\textnormal{u}^*}(r)\ge \tau_l^{\textnormal{u}}(r) = \frac{r(M-1)+K}{K}.
\end{equation}
The lower bound of NULT in Theorem \ref{lower_bound} is thus proved.
\subsubsection{Optimality}Since (\ref{lowerbound_tauu}) is the same as achievable bound (\ref{achieveUpoad}), the NULT in (\ref{achieveUpoad}) is optimal for any given $r$, or more sufficiently, for any fixed $\left(\tau^\text{c}(r,q),\tau^\text{d}(r,q)\right)$, as stated in Lemma \ref{gapLemma}.
\vspace{-4mm}
\subsection{Lower Bound and Multiplicative Gap Analysis of NCT}\label{binaryconverseup}
\subsubsection{Lower bound} Let $\{X_k\}_{q:K}$ denote the $q$-th smallest value of $K$ variables $\{X_k\}^{K}_{k=1}$ and $q\!:\!K$ denote the index of $q$-th smallest variable. For a particular task assignment policy $\big\{\mathcal{U}_{i,\mathcal{K}^{'}}\!\big\}$ with repetition order $r$ and recovery order $q$ and satisfying (\ref{cons111}) and (\ref{cons222}), the computation time when the $q$-th fastest EN finishes its assigned tasks is lower bounded by
\begin{align}
T_{q:K}^{\textnormal{c}} &= \bigg\{\!\mu m\sum\limits_{i\in\mathcal{M}}\gamma_{i,k}N \omega_k\!\bigg\}_{\!\!q:K}\nonumber\\
&\stackrel{(g)}{\ge}  \max\limits_{t\in[q]}\bigg\{\!\mu m \bigg\{\sum\limits_{i\in\mathcal{M}}\gamma_{i,k}N \bigg\}_{\!\!t:K}\!\!\!\cdot\omega_{q-t+1:K}\!\bigg\},
\end{align}
where $(g)$ follows the fact that for 2 sequences $\{x_k\}_{k=1}^{K}$ and $\{y_k\}_{k=1}^{K}$, given $\{x_k\}_{t:K}\{y_k\}_{q-t+1:K}$ for $t\!\in\![q]$, there are at most $q\!-\!1$ product values among $\{x_ky_k\}^K_{k=1}$ less than $\{x_k\}_{t:K}\{y_k\}_{q-t+1:K}$ for $t\!\in\![q]$. So the $q$-th smallest product satisfies $\{x_ky_k\}_{q:K}\!\ge\!\{x_k\}_{t:K}\{y_k\}_{q-t+1:K}$ for $t\!\in\![q]$. Taking the expectation on $T_{q:K}^{\textnormal{c}}$, we have
\begin{align}
\mathbb{E}\left[ T_{q:K}^{\textnormal{c}} \right] & \ge \mathbb{E}\left[\max\limits_{t\in[q]}\bigg\{\!\mu m \bigg\{\sum\limits_{i\in\mathcal{M}}\gamma_{i,k}N \bigg\}_{\!\!t:K}\!\!\!\cdot\omega_{q-t+1:K}\!\bigg\}\right]\nonumber
\end{align}
\begin{align}
&\quad\quad\quad\quad\stackrel{(h)}{\ge} \max\limits_{t\in[q]}\bigg\{\!\mu m \bigg\{\sum\limits_{i\in\mathcal{M}}\gamma_{i,k}N \bigg\}_{\!\!t:K}\!\!\!\cdot\mathbb{E}\left[\omega_{q-t+1:K}\right]\!\bigg\} \nonumber\\
&\quad\quad\quad\quad\stackrel{(i)}{=} \max\limits_{t\in[q]}\frac{(H_K - H_{K-q+t-1})\mu m}{\eta} \left\{\sum\limits_{i\in\mathcal{M}}\gamma_{i,k}N \right\}_{\!\!t:K},
\end{align}
where (h) follows $\mathbb{E}\big[\max\limits_{t}x_t\big]\!\ge\!\max\limits_{t}\mathbb{E}\left[x_t\right]$, (i) uses the $(q\!-\!t\!+\!1)$-th order statistic of $K$ i.i.d exponential random variables. The second term denotes the $t$-th smallest value among $K$ EN workload sizes. By (\ref{cons111}) and (\ref{cons222}), for $\forall i\!\in\!\mathcal{M}$, we let $\gamma_{i,k}\!=\!1$, $k\!=\!t\!+\!1\!:\!K,t\!+\!2\!:\!K,\cdots,K\!:\!K$, where $k\!=\!t\!:\!K$ denotes the index of the $t$-th smallest value in $\{\gamma_{i,k}\!\}_{\!k\in[K]}$. So the sum of the $t$ smallest values ($k\!=\!1\!:\!K,\cdots,t\!:\!K$) is lower bounded by $(r\!-\!K\!+t)^{+}NM$. Since the second term also represents the largest value among those $t$ smallest EN workload sizes, so this term can be further lower bounded by the average value $(r\!-\!K\!+ t)^{+}NM/t$. So the average time for the $q$ fastest ENs to finish their tasks is lower bounded by
$T^\textnormal{c}(r,q)\!\ge\!\max\limits_{t\in[q]}\big((H_K - H_{K-q+t-1})\mu m/\eta\big)\big((r-\!K\!+t)^{+}NM/t\big)$. Normalizing it by $Nm/\eta$, the lower bound of the minimum NCT is given by
\begin{equation}
\tau^{\textnormal{c}^*}(r,q)\ge\tau_l^{\textnormal{c}}(r,q)=\max\limits_{t\in[q]}\frac{(H_K - H_{K-q+t-1})(r-\!K\!+t)^{+}M\mu }{t}. \label{low_compu}
\end{equation}
\vspace{-4mm}
\subsubsection{Multiplicative gap} The multiplicative gap between the achievable NCT in Theorem \ref{achievableresults} and the lower bound (\ref{low_compu}) satisfies
\begin{align}
\frac{\tau_{\textnormal{a}}^{\textnormal{c}}(r,q)}{\tau_l^{\textnormal{c}}(r,q)} &\le \min\limits_{t\in[q]}\frac{Mr\mu(H_k\!-\!H_{K-q})t}{K(H_K - H_{K-q+t-1})(r-\!K\!+t)^{+}M\mu}\nonumber
\end{align}
\begin{align}
~~\quad\quad\quad\quad&\le \min\limits_{t\in[q]} \frac{t}{(r-\!K\!+t)^{+}} \cdot\left(1+\frac{H_{K-q+t-1}\!-\!H_{K-q}}{H_K - H_{K-q+t-1}}\right)\nonumber\\
&\le \frac{q/2}{(r-\!K\!+q/2)^{+}} \cdot \left(1+\frac{\frac{q}{2}\frac{1}{K-q+1}}{\frac{q}{2}\frac{1}{K}}\right) \nonumber\\
&= \frac{q/2}{(r-\!K\!+q/2)^{+}} \cdot \left(1+\frac{K}{K-q+1}\right).
\end{align}
When $r\!\ge\!K\!-\!n_1$ and $q\!\le\!K(1\!-\!1/n_2)\!+\!1$ with integers $0\!\le\!n_1\!<\!q/2$ and $n_2\!\ge\!1$, we have $\frac{q/2}{(r-\!K\!+q/2)^{+}}\!\le\!\frac{q/2}{q/2-n_1}\!\le\!n_1\!+\!1$ and $K/(K\!-q+\!1)\!\le\!n_2$, respectively, and consequently, we have $\tau_{\textnormal{a}}^{\textnormal{c}}/\tau_l^{\textnormal{c}}\!\le\!(1\!+n_1 )(1\!+n_2)$. Since the NULT is optimal and increases strictly with $r$, the repetition order satisfies $r\!\ge\!\!K\!-n_1$ when the NULT $\tau^\textnormal{u}\!\ge\!\tau_{\textnormal{a}}^{\textnormal{u}}(K\!-\!n_1)$.

We thus prove the lower bound of NCT in (\ref{lowercompute}) and the order-optimality of NCT in (\ref{gapnct}). \vspace{-2mm}
\subsection{Lower Bound and Multiplicative Gap Analysis of NDLT}\label{binaryconverseup}
\subsubsection{Lower bound}
\begin{figure}[t]
\centering
\includegraphics[width=2.7in, height=2.2in]{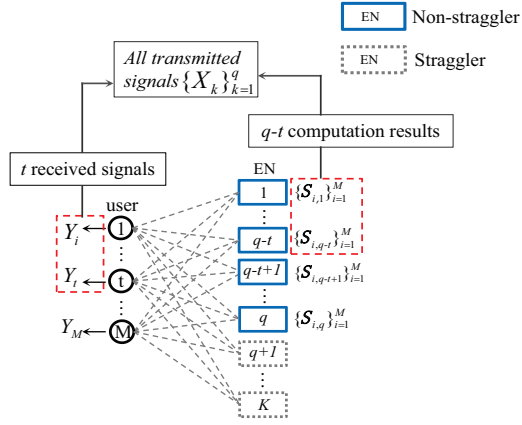}
\vspace{-2mm}
\caption{Illustration of the converse proof for NDLT.}
\vspace{-8mm}
\label{conversefigure}
\end{figure}
For a particular task assignment policy $\big\{\mathcal{U}_{i,\mathcal{K}^{'}}\!\big\}$ satisfying (\ref{cons111}) and (\ref{cons222}), and a particular subset of $q$ ENs denoted as $\mathcal{K}_q\!\subseteq\!\mathcal{K}$ whose outputs are available, each EN $k\!\in\!\mathcal{K}_q$ is assigned $r_{i,k}N$ input vectors from each user $i\!\in\!\mathcal{M}$ and can store $\mu m$ rows of $\mathbf{A}$. Since each user $i$ wants $mN$ row-vector product results $\{\mathbf{v}_{i,j}\!=\!\mathbf{A}_{m\times n}\mathbf{u}_{i,j}\}_{j\in[N]}$, it is equivalent to state that each EN $k$ can store $r_{i,k}\mu$ fractional outputs desired by each user $i$, denoted as $\mathcal{S}_{i,k}\!\triangleq\!\{\mathbf{A}_k\mathbf{u}_{i,j}\!:\mathbf{u}_{i,j}\!\in\!\mathcal{U}_{i,\mathcal{K}^{'}},k\!\in\!\mathcal{K}^{'}\}$ and with size $|\mathcal{S}_{i,k}|\!=\!\gamma_{i,k}\mu NmB$ bits, where $\gamma_{i,k}$ satisfies (\ref{cons111}) and (\ref{cons222}).
Thus, the policy $\big\{\mathcal{U}_{i,\mathcal{K}^{'}}\!\big\}_{i\in\mathcal{M},k\in\mathcal{K}}$ with an available EN set $\mathcal{K}_q$ is equivalent to a particular computation results distribution $\{\mathcal{S}_{i,k}\}_{i\in\mathcal{M},k\in\mathcal{K}_q}$.

Let $\mathcal{M}_t\!\subseteq\!\mathcal{M}$ denote an arbitrary subset of $t$ users and $\mathcal{Q}_{q-t}\!\subseteq\!\mathcal{K}_q$ denote an arbitrary subset of $q\!-t$ ENs. Also, we have $\mathcal{M}_{M-t}\!=\!\mathcal{M}/\mathcal{M}_t$ and $\mathcal{Q}_t\!=\!\mathcal{K}_q/\mathcal{Q}_{q-t}$. For a particular computation results distribution $\{\mathcal{S}_{i,k}\}_{i\in\mathcal{M},k\in\mathcal{K}_q}$, we adopt the arguments proved in \cite[Lemma 6]{sengupta2017fog} to derive the lower bound of the NDLT, i.e., intuitively, as shown in Fig. \ref{conversefigure}, \emph{given any subset of $t$ signals received at $t\!\le\!\min\{q,M\}$ users, denoted as $\{Y_i\}_{i\in\mathcal{M}_t}$, and the stored computation results information of $q-t$ ENs, denoted as $\{\mathcal{S}_{i,k}\}_{i\in\mathcal{M},k\in\mathcal{Q}_{q-t}}$, all transmitted signals $\{X_{k}\}_{k\in\mathcal{K}_q}$ and all the desired outputs $\{\mathbf{v}_{i,j}\}_{i\in\mathcal{M},j\in[N]}$ can be resolved in the high-SNR regime.} First, we have the following equality,
\begin{align}
&MNmB \!= \! H\!\left(\{\mathbf{v}_{i,j}\}_{i\in\mathcal{M},j\in[N]}\right) \nonumber\\
\!&= \! I\!\left(\!\{\mathbf{v}_{i,j}\}_{i\in\mathcal{M},j\in[N]};\{Y_i\}_{i\in\mathcal{M}_t},\{\mathcal{S}_{i,k}\}_{i\in\mathcal{M},k\in\mathcal{Q}_{q\!-\!t}}\!\right) \!+ \! H\!\left(\!\{\mathbf{v}_{i,j}\}_{i\in\mathcal{M},j\in[N]}|\{Y_i\}_{i\in\mathcal{M}_t},\{\mathcal{S}_{i,k}\}_{i\in\mathcal{M},k\in\mathcal{Q}_{q\!-\!t}}\!\right)\!.\label{ineq1}
\end{align}
For the first term, following steps in \cite[Eq. (64)]{sengupta2017fog}, we have
\begin{small}\begin{align}
&I\left(\{\mathbf{v}_{i,j}\}_{i\in\mathcal{M},j\in[N]};\{Y_i\}_{i\in\mathcal{M}_t},\{\mathcal{S}_{i,k}\}_{i\in\mathcal{M},k\in\mathcal{Q}_{q\!-\!t}}\right) \nonumber\\
&=I\left(\{\mathbf{v}_{i,j}\}_{i\in\mathcal{M},j\in[N]};\{Y_i\}_{i\in\mathcal{M}_t}\right) + I\left(\{\mathbf{v}_{i,j}\}_{i\in\mathcal{M},j\in[N]};\{\mathcal{S}_{i,k}\}_{i\in\mathcal{M},k\in\mathcal{Q}_{q\!-\!t}}|\{Y_i\}_{i\in\mathcal{M}_t}\right)\nonumber\\
&\le I\left(\{\mathbf{v}_{i,j}\}_{i\in\mathcal{M},j\in[N]};\{Y_i\}_{i\in\mathcal{M}_t}\right) +
I\left(\{\mathbf{v}_{i,j}\}_{i\in\mathcal{M},j\in[N]};\{\mathcal{S}_{i,k}\}_{i\in\mathcal{M},k\in\mathcal{Q}_{q\!-\!t}},\{\mathbf{v}_{i,j}\}_{i\in\mathcal{M}_t,j\in[N]}|\{Y_i\}_{i\in\mathcal{M}_t}\right)~~~~~~~~~\nonumber\\
&=I\left(\{\mathbf{v}_{i,j}\}_{i\in\mathcal{M},j\in[N]};\{Y_i\}_{i\in\mathcal{M}_t}\right) +
I\left(\{\mathbf{v}_{i,j}\}_{i\in\mathcal{M},j\in[N]};\{\mathbf{v}_{i,j}\}_{i\in\mathcal{M}_t,j\in[N]}|\{Y_i\}_{i\in\mathcal{M}_t}\right)
\nonumber\\
&~~~~~~~~~~~~~~~~~~~~~~~~~~~~~~~~~~~~~~~~+\!I\left(\{\mathbf{v}_{i,j}\}_{i\in\mathcal{M},j\in[N]};\{\mathcal{S}_{i,k}\}_{i\in\mathcal{M},k\in\mathcal{Q}_{q\!-\!t}},|\{\mathbf{v}_{i,j}\}_{i\in\mathcal{M}_t,j\in[N]},\{Y_i\}_{i\in\mathcal{M}_t}\right)\nonumber\\
&\le I\left(\{\mathbf{v}_{i,j}\}_{i\in\mathcal{M},j\in[N]};\{Y_i\}_{i\in\mathcal{M}_t}\right) +
H\left(\{\mathbf{v}_{i,j}\}_{i\in\mathcal{M}_t,j\in[N]}|\{Y_i\}_{i\in\mathcal{M}_t}\right)
\nonumber\\
&~~~+H\left(\{\mathcal{S}_{i,k}\}_{i\in\mathcal{M},k\in\mathcal{Q}_{q\!-\!t}}|\{\mathbf{v}_{i,j}\}_{i\in\mathcal{M}_t,j\in[N]},\{Y_i\}_{i\in\mathcal{M}_t}\right)
-H\left(\{\mathcal{S}_{i,k}\}_{i\in\mathcal{M},k\in\mathcal{Q}_{q\!-\!t}}|\{\mathbf{v}_{i,j}\}_{i\in\mathcal{M},j\in[N]},\{Y_i\}_{i\in\mathcal{M}_t}\right)\nonumber\\
&\!\stackrel{(j)}{\le}\!h\left(\{Y_i\}_{i\in\mathcal{M}_t}\right)-h\left(\{Y_i\}_{i\in\mathcal{M}_t}|\{\mathbf{v}_{i,j}\}_{i\in\mathcal{M},j\in[N]}\right) +
tNmB\epsilon  + H\left(\{\mathcal{S}_{i,k}\}_{i\in\mathcal{M},k\in\mathcal{Q}_{q\!-\!t}}|\{\mathbf{v}_{i,j}\}_{i\in\mathcal{M}_t,j\in[N]}\right)\nonumber\\
& \stackrel{(k)}{\le} tT\log\left(2\pi e(\Lambda P^\textnormal{d}+1)\right)\!-\! h\left(\{n_i\}_{i\in\mathcal{M}_t}|\{\mathbf{v}_{i,j}\}_{i\in\mathcal{M},j\in[N]}\right)\!+
tNmB\epsilon  +\! \sum\limits_{k\in\mathcal{Q}_{q\!-\!t}}\!\!H\left(\{\mathcal{S}_{i,k}\}_{i\in\mathcal{M}}|\{\mathbf{v}_{i,j}\}_{i\in\mathcal{M}_t,j\in[N]}\right)\nonumber\\
& \stackrel{(l)}{\le}\!tT\log\left(2\pi e(\Lambda P^\textnormal{d}+1)\right) -  tT\log\left(2\pi e \right) +
tNmB\epsilon  + \sum\limits_{k\in\mathcal{Q}_{q\!-\!t}}\sum\limits_{i\in\mathcal{M}_{M\!-\!t}}H\left(\mathcal{S}_{i,k}\right)\nonumber\\
&\le tT\log\left(\Lambda P^\textnormal{d}+1\right) + tNmB\epsilon  + \sum\limits_{k\in\mathcal{Q}_{q\!-\!t}}\sum\limits_{i\in\mathcal{M}_{M\!-\!t}}\gamma_{i,k}\mu NmB, \label{ineq2}
 \end{align}\end{small}
\!\!where, in step $(j)$, $\{Y_i\}$ are continuous random variables, the third term uses Fano{'}s inequality, the fourth term is because dropping the condition increases the entropy, the last term is 0 since the storage information $\{\mathcal{S}_{i,k}\}$ are the functions of $\{\mathbf{v}_{i,j}\}_{i\in\mathcal{M},j\in[N]}$; In step $(k)$, the first term uses \cite[Lemma 5]{sengupta2017fog}, and note that $\Lambda$ defined in \cite[Lemma 5]{sengupta2017fog} is a constant only depending on downlink channel coefficients in $\mathbf{H}^\textnormal{d}$.
For the second term, by \cite[Lemma 6]{sengupta2017fog} that proves the adopted argument, we have
\begin{align}
\!\!\!H\left(\{\mathbf{v}_{i,j}\}_{i\in\mathcal{M},j\in[N]}|\{Y_i\}_{i\in\mathcal{M}_t},\{\mathcal{S}_{i,k}\}_{i\in\mathcal{M},k\in\mathcal{Q}_{q\!-\!t}}\right)
\!\le\!tNmB\epsilon\!+\!T\log\det\left( \mathbf{I}_{M-t} +\tilde{\mathbf{H}}^\textnormal{d}(\tilde{\mathbf{H}}^\textnormal{d})^{H}\right)\!,\!\! \label{ineq3}
\end{align}
where the $(M\!-\!t)\!\times\!(M\!-\!t)$ matrix $\tilde{\mathbf{H}}^\textnormal{d}$ defined in \cite[Lemma 6]{sengupta2017fog} only depends on the channel matrix $\mathbf{H}^\textnormal{d}$, and $\mathbf{I}_{M\!-\!t}$ is a $(M\!-\!t)\!\times\!(M\!-\!t)$ identity matrix. The expressions of $\tilde{\mathbf{G}}$ and $\Lambda$ are omitted here since they can be treated as constants.

Substituting (\ref{ineq2}) and (\ref{ineq3}) into (\ref{ineq1}), we have
\begin{align}
\!\!\!M NmB\!\le\!tT\log\left(\Lambda P^\textnormal{d}+1\right)\! +\! 2tNmB\epsilon \! +\!\!\sum\limits_{k\in\mathcal{Q}_{q\!-\!t}}\sum\limits_{i\in\mathcal{M}_{M\!-\!t}}\!\!\! \! \gamma_{i,k}\mu NmB\!+\! T\log\det\! \left( \mathbf{I}_{M-t}\!+\!\tilde{\mathbf{H}}^\textnormal{d}(\tilde{\mathbf{H}}^\textnormal{d})^{H}\right)\!,\!\!
\end{align}
Moving $T$ to the left side and dividing by $\frac{NmB}{\log P^\textnormal{d}}$, we have
\begin{align}
\frac{T}{NmB/\log P^\textnormal{d}}\ge \frac{M\!-\!\sum\limits_{k\in\mathcal{Q}_{q\!-\!t}}\sum\limits_{i\in\mathcal{M}_{M\!-\!t}}\!\!\gamma_{i,k}\mu\!-\!2t\epsilon}{t}
\cdot\!\frac{t\log P^\textnormal{d}}{t\log\left(\Lambda P^\textnormal{d}\!+\!1\right)\!+\!\log\det\left( \mathbf{I}_{M-t} \!+\!\tilde{\mathbf{H}}^\textnormal{d}(\tilde{\mathbf{H}}^\textnormal{d})^{H}\right)}.
\end{align}
Taking $P^\textnormal{d}\to\infty$ and $\epsilon\to 0$ as $B\to\infty$, the minimum NDLT under the output distribution  $\{\mathcal{S}_{i,k}\}_{i\in\mathcal{M},k\in\mathcal{Q}_{q\!-\!t}}$ is lower bounded by
\begin{equation}
\tau^{\textnormal{d}^*}\!(r,\mathcal{K}_{q},\mathcal{Q}_{q-t}) \ge  \frac{M\!-\!\!\sum\limits_{k\in\mathcal{Q}_{q\!-\!t}}\sum\limits_{i\in\mathcal{M}_{M\!-\!t}}\!\!\!\gamma_{i,k}\mu}{t},~ \forall \mathcal{Q}_{q-t}\subseteq\mathcal{K}_q.\!\!
\end{equation}
Note that the adopted argument holds for any subset of $q\!-\!t$ ENs (see Fig. \ref{conversefigure}). Thus, by tasking the sum over all possible subset $\mathcal{Q}_{q-t}\subseteq\mathcal{K}_q$, we have
\begin{align}
\binom{q}{q\!-\!t}\tau^{\textnormal{d}^*}(r,\mathcal{K}_{q},q-t)& \ge \sum\limits_{\mathcal{Q}_{q-t}\subseteq\mathcal{K}_q} \!\!\!\frac{M\!-\!\!\sum\limits_{k\in\mathcal{Q}_{q\!-\!t}}\sum\limits_{i\in\mathcal{M}_{M\!-\!t}}\!\!\!\gamma_{i,k}\mu}{t}\nonumber\\
& = \frac{\binom{q}{q-t}M\!-\!\!\sum\limits_{i\in\mathcal{M}_{M\!-\!t}}\sum\limits_{\mathcal{Q}_{q-t}\subseteq\mathcal{K}_q}\sum\limits_{k\in\mathcal{Q}_{q\!-\!t}}\!\!\!\gamma_{i,k}\mu}{t}\nonumber\\
& = \frac{\binom{q}{q-t}M\!-\!\!\sum\limits_{i\in\mathcal{M}_{M\!-\!t}}\binom{q-1}{q-t-1}\sum\limits_{k\in\mathcal{K}_q}\!\!\gamma_{i,k}\mu}{t}.\!\!
\end{align}
For the particular policy $\big\{\mathcal{U}_{i,\mathcal{K}^{'}}\!\big\}$ with repetition order $r$ and satisfying (\ref{cons111}) and (\ref{cons222}), this lower bound also holds for any subset $\mathcal{K}_{q}$ since $K\!-\!q$ stragglers occur randomly (see Fig. \ref{conversefigure}), by taking the sum over all possible subsets $\mathcal{K}_{q}\subseteq\mathcal{K}$, we have
\begin{align}
\!\!\!\!\binom{\!K\!}{\!q\!}\!\binom{\!q\!}{\!q\!-\!t\!}\tau^{\textnormal{d}^*}\!(r,q,q\!-\!t)
& \ge  \sum\limits_{\mathcal{K}_{q}\subseteq\mathcal{K}}\!\!\!\frac{\binom{q}{q-t}M\!-\!\!\sum\limits_{i\in\mathcal{M}_{\!M\!-\!t}}\!\!\!\!\binom{q-1}{q-t-1}\!\!\sum\limits_{k\in\mathcal{K}_q}\!\!\gamma_{i,k}\mu}{t}\nonumber\\
& = \frac{\binom{K}{q}\binom{q}{q-t}M\!-\!\!\!\sum\limits_{i\in\mathcal{M}_{\!M\!-\!t}}\!\!\!\binom{q-1}{q-t-1}\!\sum\limits_{\mathcal{K}_{q}\subseteq\mathcal{K}}\sum\limits_{k\in\mathcal{K}_q}\!\!\gamma_{i,k}\mu}{t}\nonumber\\
& = \frac{\binom{K}{q}\binom{q}{q-t}M\!-\!\!\!\sum\limits_{i\in\mathcal{M}_{M\!-\!t}}\!\!\!\binom{q-1}{q-t-1}\binom{K-1}{q-1}\!\!\sum\limits_{k\in\mathcal{K}}\!\!\gamma_{i,k}\mu}{t}\nonumber\\
& \stackrel{(m)}{=}\frac{\binom{K}{q}\binom{q}{q-t}M\!\!-\!(\!M\!\!-\!t)\binom{q-1}{q-t-1}\binom{K-1}{q-1}r\mu}{t},\!\!\label{eq46}
\end{align}
where $(m)$ is due to (\ref{cons111}). Remanaging (\ref{eq46}), the lower bound of NDLT at pair $(r,q)$ is given by
\begin{align}
\tau^{\textnormal{d}^*}(r,q,q-t)&\ge \frac{M-(M\!-\!t)(q-t)\frac{r}{K}\mu}{t}, \label{tau_t}
\end{align}
Since the argument we adopt to derive (\ref{tau_t}) holds for $1\!\le\!t\!\le\!\min\{q,M\}$, the lower bound of the minimum NDLT at pair $(r,q)$ can be optimized as
\begin{equation} \label{lowerbound_taud}
\tau^{\textnormal{d}^*}(r,q)\ge \tau_l^{\textnormal{d}}(r,q)= \max_{t\in\{1,\cdots,\min\{q,M\}\}}\frac{M-(M\!-\!t)(q-t)\frac{r}{K}\mu}{t}.
\end{equation}
\vspace{-3mm}
\subsubsection{Multiplicative gap}
By (\ref{achievedownload}), the achievable NDLT is upper bounded by
\begin{align}
\tau_{\textnormal{a}}^{\textnormal{d}}&=\sum\limits^{\min\{r,q\}}_{p_1=r-K+q}B_{p_1}\left(\sum\limits^{l_\textnormal{max}}_{p_2=l_\textnormal{min}}\frac{B_{p_2}}{d_{p_1,M,p_2}^{\textnormal{d}}}\!+\!\frac{B_{l_{p_1}\!-\!1}}{d^\textnormal{d}_{p_1,M,l_{p_1}\!-\!1}}\right)\nonumber\\
&\stackrel{(m)}{\le}\sum\limits^{\min\{r,q\}}_{p_1=r-K+q}B_{p_1}\frac{\sum\limits^{l_\textnormal{max}}_{p_2=l_\textnormal{min}}\!B_{p_2}\!+\!B_{l_{p_1}\!-\!1}}{d^\textnormal{d}_{p_1,M,1}}\nonumber\\
&\stackrel{(n)}{\le}\frac{1}{d^\textnormal{d}_{r-K+q,M,1}},
\end{align}
where $(m)$ is because $d_{p_1,M,p_2}^{\textnormal{d}}$ increases with $p_2$ \cite[Lemma 1]{gckkl} and $(n)$ is because $d^\textnormal{d}_{p_1,M,1}\!=\!p_1/(p_1\!+\!M\!-\!1)$ increases with $p_1$. By (\ref{lowerbound_taud}), we have $\tau_l^{\textnormal{d}}(r,q)\!\ge\!M/\min\{q,M\}$, so the multiplicative gap satisfies
\begin{align}
\frac{\tau_{\textnormal{a}}^{\textnormal{d}}}{\tau_l^{\textnormal{d}}}&\le\frac{\min\{q,M\}}{d^\textnormal{d}_{r-K+q,M,1}} =\frac{\min\{q,M\}(r\!-\!K\!+\!q\!+\!M\!-\!1)}{(r\!-\!K\!+\!q)M}.
\end{align}
If $q\!\le\!M$, we have $\tau_{\textnormal{a}}^{\textnormal{d}}/\tau_l^{\textnormal{d}}\!\le\!\frac{q}{r-K+q}(\frac{q}{M}\!+\!\frac{M-1}{M}\!-\!\frac{K-r}{M})\!\le\!\frac{2q}{r-K+q}\!\le\!\frac{2q}{q-n}\!\le\!2(n\mu\!+\!1)$ for $r\!\ge\!K\!-\!n$; otherwise, we have $\tau_{\textnormal{a}}^{\textnormal{d}}/\tau_l^{\textnormal{d}}\!\le\!1\!+\!\frac{M-1}{r-K+q}\!\le\!1\!+\!\frac{q-1}{q-n}\!\le\!2\!+\!(n\!-\!1)\mu$ for $r\!\ge\!K\!-\!n$. Here, integer $n$ satisfies $n\!\le\!q\!-\!1/\mu$ due to $(r\!-\!K\!+\!q)\mu\!\ge\!1$. In summary, since $2(n\mu\!+\!1)\!>\!2\!+\!(n\!-\!1)\mu$, we have $\tau_{\textnormal{a}}^{\textnormal{d}}/\tau_l^{\textnormal{d}}\!\le\!2(n\mu\!+\!1)$ for $r\!\ge\!K\!-\!n$. Furthermore, when the NULT $\tau^\textnormal{u}\!\ge\!\tau_{\textnormal{a}}^{\textnormal{u}}(K\!-n)$, the repetition order satisfies $r\!\ge\!K\!-\!n$. Thus, when $r\!=\!K$, or equivalently, $\tau^\textnormal{u}\!\ge\!\tau_{\textnormal{a}}^{\textnormal{u}}(K)$, we have $\tau_{\textnormal{a}}^{\textnormal{d}}/\tau_l^{\textnormal{d}}\!\le\!2$.

Next, consider the special case $\mu\!=\!1$. For any input vectors with degrees of replication of $p_1$, the degrees of replication of their associated outputs is $p_2\!=\!l_\textnormal{max}\!=\!l_\textnormal{min}\!=\!p_1$, and we also have $B_{p_2}\!=\!1$, so the achievable NDLT in (\ref{achievedownload}) can be simplified as
\begin{align}
\tau_{\textnormal{a}}^{\textnormal{d}}=\sum\limits^{\min\{r,q\}}_{p_1=r-K+q}B_{p_1}\frac{1}{d^\textnormal{d}_{p_1,M,p_1}}
\le M\frac{\sum\limits^{\min\{r,q\}}_{p_1=r\!-\!K\!+q}B_{p_1}}{\min\{r-K+q,M\}}
=\frac{M}{\min\{r\!-\!K\!+q,M\}}.
\end{align}
Due to $\tau_l^{\textnormal{d}}(r,q)\!\ge\!M/\min\{q,M\}$, we have
$\tau_{\textnormal{a}}^{\textnormal{d}}/\tau_l^{\textnormal{d}}\!\le\!\min\{q,M\}/\min\{r\!-\!K\!+\!q,M\}.$
It is seen that when $\tau^\textnormal{u}(r)\!\ge\!\tau_{\textnormal{a}}^{\textnormal{u}}(M\!+\!K\!-q)$, we have $\tau_{\textnormal{a}}^{\textnormal{d}}\!=\tau_l^{\textnormal{d}}\!=\!1$ that is optimal; when $\tau_{\textnormal{a}}^{\textnormal{u}}(K\!-\!n)\!\le\!\tau^\textnormal{u}(r)\!<\!\tau_{\textnormal{a}}^{\textnormal{u}}(M\!+\!K\!-\!q)$, we have $\tau_{\textnormal{a}}^{\textnormal{d}}/\tau_l^{\textnormal{d}}\!\le\!q/(r\!-\!K\!+\!q)\!\le\!q/(q\!-\!n)\!\le\!n\!+\!1$, where the integer $n$ satisfies $\!q\!-\!M\!<\!n\!\le\!q\!-\!1$. We prove the lower bound of NDLT in (\ref{lowerdownload}) and the multiplicative gap in (\ref{gapndlt}).

\vspace{-2mm}
\subsection{Outer Bound of Compute-Download Latency Region}
Based on the feasible set $\mathcal{R}$ in (\ref{regionR}) and the convexity of $\mathscr{T}^{*}(\tau^\textnormal{u})$ in Remark \ref{remarkconvex}, for an NULT $\tau^\textnormal{u}\!=\!\tau_{\textnormal{a}}^{\textnormal{u}}(r)$ in (\ref{achieveUpoad}) for some $r$, an outer bound $\mathscr{T}_{out}(\tau^\textnormal{u})$ of the compute-download latency region is given as the convex hull of set
$\big\{\!\!\left(\tau_l^{\textnormal{c}}(r,q),\tau_l^{\textnormal{d}}(r,q)\right)\!\!:\!q\!\in\!\!\big[\lceil\!\frac{1}{\mu}\!\rceil\!+\!K\!-r\!:\!\!K\big]\!\big\}$.

\bibliographystyle{IEEEtran}
\bibliography{refer}

\end{spacing}

\end{document}